\documentclass[letterpaper,11pt]{article}
\usepackage[utf8]{inputenc}
\usepackage{latexsym}
\usepackage{amssymb}
\usepackage[margin=1in]{geometry}
\usepackage{amsmath,amsfonts,amsthm}
\usepackage{dsfont}
\usepackage{graphicx}
\usepackage{multirow}
\usepackage{ifsym}
\usepackage{color}
\usepackage{dsfont}
\usepackage{xcolor}
\usepackage{algorithm}
\usepackage{algorithmicx}
\usepackage{algpseudocode}
\usepackage{natbib}
\usepackage{hyperref}
\hypersetup{%
   breaklinks,%
   colorlinks=true,%
   linkcolor=black,%
   urlcolor=black,%
   citecolor=[rgb]{0,0,0.45}
}

\theoremstyle{plain}

\theoremstyle{plain}
\newtheorem{theorem}{Theorem}[section]
\newtheorem{lemma}[theorem]{Lemma}
\newtheorem{claim}[theorem]{Claim}

\newtheorem{definition}[theorem]{Definition}
\theoremstyle{definition}
\newtheorem{observation}[theorem]{Observation}
\newtheorem{remark}[theorem]{Remark}
\newtheorem{corollary}[theorem]{Corollary}

\DeclareSymbolFont{bbold}{U}{bbold}{m}{n}
\DeclareSymbolFontAlphabet{\mathbbold}{bbold}

\newcommand{\nats}{{\mathbb N}}

\newcommand{\eps}{{\varepsilon}}

\newcommand{\prob}{\mathrm{Prob}}

\newcommand{\expec}{\mathbb{E}}

\newcommand{\polylog}{\mathrm{polylog}}

\newcommand{\opt}{\mathrm{OPT}}
\newcommand{\alg}{\mathrm{ALG}}

\title{Approximating Red-Blue Set Cover and Minimum Monotone Satisfying Assignment}
\author{Eden Chlamt\'a\v{c}\thanks{Department of Computer Science, Ben-Gurion University. The work was done while the author was visiting and supported by TTIC. Email: \href{chlamtac@cs.bgu.ac.il}{chlamtac@cs.bgu.ac.il}} \and Yury Makarychev\thanks{Toyota Technological Institute at Chicago (TTIC). Supported by NSF awards CCF-1955173 and CCF-1934843. Email: \href{mailto:yury@ttic.edu}{yury@ttic.edu}} \and Ali Vakilian\thanks{Toyota Technological Institute at Chicago (TTIC). Email: \href{mailto:vakilian@ttic.edu}{vakilian@ttic.edu}}}
\date{}
\begin{document}

\maketitle
\begin{abstract}
We provide new approximation algorithms for the Red-Blue Set Cover and Circuit Minimum Monotone Satisfying Assignment (MMSA) problems. Our algorithm for Red-Blue Set Cover achieves $\tilde O(m^{1/3})$-approximation improving on the $\tilde O(m^{1/2})$-approximation due to Elkin and Peleg (where $m$ is the number of sets). Our approximation algorithm for MMSA$_t$ (for circuits of depth $t$) gives an $\tilde O(N^{1-\delta})$ approximation for $\delta = \frac{1}{3}2^{3-\lceil t/2\rceil}$, where %$t\geq 4$ is the depth of the circuit and 
$N$ is the number of gates and variables. No non-trivial approximation algorithms for MMSA$_t$ with $t\geq 4$ were previously known.

We complement these results with lower bounds for these problems: For Red-Blue Set Cover,  we provide a nearly approximation preserving reduction from Min $k$-Union that gives an $\tilde\Omega(m^{1/4 - \varepsilon})$ hardness under the Dense-vs-Random conjecture, while for MMSA we sketch a proof that an SDP relaxation strengthened by Sherali--Adams has an integrality gap of $N^{1-\varepsilon}$ where $\varepsilon \to 0$ as the circuit depth $t\to \infty$.

%Additionally, we provide a nearly approximation preserving reduction from Min $k$-Union to Red-Blue Set Cover that gives an $\tilde\Omega(m^{1/4 - \varepsilon})$ hardness under the Dense-vs-Random conjecture. We also sketch a proof that an SDP relaxation for MMSA strengthened by Sherali--Adams has an integrality gap of $N^{1-\varepsilon}$ where $\varepsilon \to 0$ as the circuit depth $t\to \infty$.
\end{abstract}

\section{Introduction}
In this paper, we study two problems, {\em Red-Blue Set Cover} and its generalization \textit{Circuit Minimum Monotone Satisfying Assignment}. Red-Blue Set Cover, a natural generalization of Set Cover, was introduced by~\cite{carr2000red}.  Circuit Minimum Monotone Satisfying Assignment, a problem more closely related to Label Cover, was introduced by~\cite{ABMP01} and~\cite{GM97}. 
\begin{definition}\label{def:rb-set-cover}
    In Red-Blue Set Cover, we are given a universe of $(k+n)$ elements $U$ partitioned into disjoint sets of red elements ($R$) of size $n$ and blue elements ($B$) of size $k$, that is $U= R \cup B$ and $R\cap B = \emptyset$, and a collection of sets $\mathcal{S} := \{S_1, \cdots, S_m\}$. The goal is to find a sub-collection of sets $\mathcal{F} \subseteq \mathcal{S}$ such that the union of the sets in $\mathcal{F}$ covers all blue elements while minimizing the number of covered red elements. 
    
    Besides Red-Blue Set Cover, we consider the {\em Partial Red-Blue Set Cover} problem in which we are additionally given a parameter $\hat k$, and the goal is cover at least $\hat k$ blue elements while minimizing the number of covered red elements. 
\end{definition}

\begin{definition}\label{def:MMSA}
The Circuit  Minimum Monotone Satisfying Assignment problem of depth $t$, denoted as MMSA$_t$, is as follows.
We are given a circuit $C$ of depth $t$ over Boolean variables $x_1,\dots, x_n$. Circuit $C$ has AND and OR gates: all gates at even distances from the root (including the output gate at the root) are AND gates; all gates at odd distances are OR gates. The goal is to find a satisfying assignment with the minimum number of variables $x_i$ set to 1 (true). 
\end{definition}
Note that $C$ computes a monotone function and the assignment of all ones is always a feasible solution.
Though the definitions of the problems are quite different, Red-Blue Set Cover and MMSA$_t$ are closely related. Namely, Red-Blue Set Cover is equivalent to MMSA$_3$.\footnote{Also observe that MMSA$_2$ is equivalent to Set Cover.} The correspondence is as follows: variables $x_1,\dots, x_n$ represent red elements; AND gates in the third layer represent sets $S_1,\dots, S_m$;  OR gates in the second layer represent %sets $S_j$; finally, AND gates in the first layer represent
blue elements. The gate for a set $S_j$ is connected to OR gates representing blue elements of $S_j$ and variables $x_i$ representing % blue and
red elements of $S_j$. %, respectively.
It is easy to see that an assignment to $x_1,\dots, x_n$ satisfies the circuit if and only if there exists a sub-family ${\cal F} \subseteq {\cal S}$ that covers %(1)
all the blue elements, and only covers red elements corresponding to variables $x_i$ which are assigned $1$ (but not necessarily all of them).%(2) all red elements whose corresponding variable $x_i$ is assigned $1$ ($\cal F$ may cover other red elements as well).

\paragraph{Background.}
Red-Blue Set Cover and its variants are related to several well-known problems in combinatorial optimization including {\em group Steiner} and {\em directed Steiner} problems, {\em minimum monotone satisfying assignment} and {\em symmetric minimum label cover}.  Arguably, the interest to the general MMSA$_t$ problem is mostly motivated by its connection to complexity and hardness of approximation.

The Red-Blue Set Cover has applications in various settings such as {\em anomaly detection, information retrieval} and notably in {\em learning of disjunctions}~\citep{carr2000red}. 
Learning of disjunctions over $\{0,1\}^m$ is one of the basic problems in the PAC model. In this problem, given an arbitrary distribution $\mathcal{D}$ over $\{0,1\}^m$ and a target function $h^*: \{0,1\}^m\rightarrow \{-1,+1\}$ which denotes the true labels of examples, the goal is to find the best disjunction $f^*: \{0,1\}^m \rightarrow \{-1, +1\}$ with respect to $\mathcal{D}$ and $h^*$ (i.e., $f^*(x)$ computes a disjunction of a subset of coordinates of $x$). The problem of learning disjunctions can be formulated as an instance of the (Partial) Red-Blue Set Cover problem~\citep{awasthi2010improved}: we can think of the positive examples as blue elements (i.e., $h^*(x)=1$) and the negative examples as red elements (i.e., $h^*(x)=-1$). Then, we construct a set $S_i$ corresponding to each coordinate $i$ and the set $S_i$ contains an example $x$ if the $i$-th coordinate of $x$ is equal to $1$. Let $C\subset \{S_1,\cdots, S_m\}$. Then, the disjunction $f_C$ corresponding to $C$, i.e., $f_C(x) := \bigvee_{S_i \in C} x_i$, outputs one on an example $x$ if in the constructed Red-Blue Set Cover instance, the element corresponding to $x$ is covered by sets in $C$. 
%By slightly modifying the reduction, we can make it approximation-preserving.

As we observe in Section~\ref{sec:reduction-k-union}, Red-Blue Set Cover is also related to the Min $k$-Union problem which is a generalization of Densest $k$-Subgraph~\citep{chlamtac2016densest}. In Min $k$-Union, given a collection of $m$ sets $\mathcal{S}$ and a target number of sets $k$, the goal is to pick $k$ sets from $\mathcal{S}$ whose union has the minimum cardinality. Notably, under a hardness assumption, which is an extension of the ``Dense vs Random" conjecture for Densest $k$-Subgraph to hypergraphs, Min $k$-Union cannot be approximated better than  $\tilde\Omega(m^{1/4-\varepsilon})$~\citep{chlamtavc2017minimizing}. In this paper, we prove a hardness of approximation result for Red-Blue Set Cover by constructing a reduction from Min $k$-Union to Red-Blue Set Cover.

\subsection{Related work}
\citeauthor{carr2000red}~formulated the Red-Blue Set Cover problem and presented a $2\sqrt{m}$-approximation algorithm for the problem when every set contains only one blue element. Later, \cite{elkin2007hardness} showed that it is possible to obtain a $2\sqrt{m\log (n+k)}$ approximation in the general case of the problem. 
This remained the best known upper bound for Red-Blue Set Cover prior to our work. No non-trivial algorithms for MMSA$_t$ were previously known for any $t\geq 4$.

On the hardness side, \cite{dinur2004hardness} showed that  MMSA$_3$ is hard to approximate within a factor of $O(2^{\log^{1-\epsilon} m})$ where $\epsilon = 1/\log\log^c m$ for any constant $c<1/2$, if $P\neq NP$. 
As was observed by \cite{carr2000red}, this implies a factor of $O(2^{\log^{1-\epsilon} m})$ hardness for Red-Blue Set Cover as well. The hardness result holds even for the special case of the problem where every set contains only one blue and two red elements. 

%Moreover, via a reduction to symmetric label cover, \cite{dodis1999design} showed that the problem is hard to approximate within a factor of $2^{\log^{1-\epsilon} m}$ for any $\epsilon>0$ unless $NP \subseteq DTIME(n^{\polylog~ n})$.

Finally, \cite{CNW16} gave a lower bound on a variant of MMSA in which the circuit depth $t$ is not fixed. Assuming a variant of the Dense-vs-Random conjecture, they showed that for every $\varepsilon > 0$, the problem does not admit an $O(n^{1/2-\varepsilon})$ approximation,  where $n$ is the number of variables, and an $O(N^{1/3-\varepsilon})$ approximation, where $N$ is the total number of gates and variables in the circuit.

\paragraph{Learning of Disjunctions.}
While algorithms for Red-Blue Set Cover return a disjunction with no error on positive examples, i.e., it covers all ``blue'' elements, it is straightforward to make those algorithms work for the case with two-sided error. 
A variant of the problem with a two-sided error is formally defined as {\em Positive–Negative Partial Set Cover}~\citep{miettinen2008positive} where the author showed that a $f(m,n)$-approximation for Red-Blue Set Cover implies a $f(m+n, n)$-approximation for Positive-Negative Partial Set Cover. Our result also holds for Partial Red-Blue Set Cover and a $c$-approximation for Partial Red-Blue Set Cover can be used to output a $c$-approximate solution of Positive-Negative Partial Set Cover. %in which the goal is to minimize the number of covered red elements subject to covering at least $\hat{k}$ blue elements for a given parameter $\hat{k}\le k$        

\cite{awasthi2010improved} designed an $O(n^{1/3+\alpha})$-approximation for any constant $\alpha>0$. While the proposed algorithm of~\citeauthor{awasthi2010improved} is an {\em agnostic learning} of disjunctions (i.e., the solution is not of form of disjunctions), employing an approximation algorithm of Red-Blue Set Cover, produces a disjunction as its output (i.e., the algorithms for Red-Blue Set Cover are {\em proper} learners). 

\paragraph{Geometric Red-Blue Set Cover.} The problem has been studied extensively in geometric models.~\cite{chan2015geometric} studied the setting in which $R$ and $B$ are sets of points in $\mathbb{R}^2$ and $\mathcal{S}$ is a collection of unit squares. They proved that the problem still remains NP-hard in this restricted instance and presented a PTAS for this problem. \cite{madireddy2022constant} designed an $O(1)$-approximation algorithm for another geometric variant of the problem, in which sets are unit disks. The problem has also been studied in higher dimensions with hyperplanes and axis-parallel objects as the sets~\citep{ashok2017multivariate,madireddy2021geometric,abidha2022red}.

\subsection{Our Results}
In this paper, we present new approximation algorithms for Red-Blue Set Cover, MMSA$_4$, and general MMSA$_{t}$. Additionally, we offer a new conditional hardness of approximation result for Red-Blue Set Cover. We also discuss the integrality gap of a basic SDP relaxation of MMSA$_{t}$ strengthened by Sherali–Adams when $t\to \infty$.

We start by describing our result for Red-Blue Set Cover.
\begin{theorem}\label{thm:main}
    There exists an $O(m^{1/3} \log^{4/3} n \log k)$-approximation algorithm for Red-Blue Set Cover where $m$ is the number of sets, $n$ is the number of red elements, and $k$ is the number of blue elements. 
\end{theorem}

As we demonstrate later, our algorithm also works for the Partial Red-Blue Set Cover. % problem in which the goal is to cover at least $\hat{k}\le k$ blue elements, too. 
Our approach partitions the instance into subinstances %is to first present an $\Tilde{O}(m^{1/3})$-approximation algorithm\footnote{Here, we abuse the $\Tilde{O}$ notation to hide polylog factors of $n,k$.} for the setting
in which all sets have a bounded number of red elements, say between $r$ and $2r$, and each red element appears in a bounded number of sets. Utilizing the properties of this partition, we show that we can always find a small collection of sets that preserves the right ratio of red to blue elements in order to make progress towards an $\Tilde{O}(m^{1/3})$-approximation algorithm.\footnote{Here, we abuse the $\Tilde{O}$ notation to hide polylog factors of $n,k$.} %Once we have an algorithm for this restricted case, we show that we can always find a subset of sets $\script{S'}$ and red elements $R'$ satisfying the required properties and the intersection of the optimal solution of the original instance and $\mathcal{R'}$ covers at least $(1/\log n)$-fraction of $B$.
Then, by applying the algorithm iteratively until all blue elements are covered, we obtain the guarantee of Theorem~\ref{thm:main}. In each iteration, our analysis guarantees the feasibility of a local % %For the restricted instance of the problem, we propose a new 
LP relaxation for which a simple randomized rounding obtains the required ratio of blue to red vertices.

Now we describe our results for the MMSA problem.
\begin{theorem}\label{thm:main-MMSA4}
    There exists an $\tilde O(N^{1/3})$-approximation algorithm for MMSA$_4$, where $N$ is the total number of gates and variables in the input instance. 
\end{theorem}

Our algorithm for MMSA$_4$ is inspired by our algorithm for MMSA$_3$, though due to the complexities of the problem, the algorithm is significantly more involved. In particular, there does not seem to be a natural preprocessing step analogous to the partition we apply for Red-Blue Set Cover, and so we need to rely on a higher-moment LP relaxation and a careful LP-based partition which is built into the algorithm.

\begin{theorem}\label{thm:main-MMSAt}
    Let $t \geq 4$. Define $\delta = \frac{1}{3} \cdot 2^{3- \lceil t/2\rceil}$. There exists an $\tilde O(N^{1- \delta})$-approximation algorithm for MMSA$_t$ where $N$ is the total number of gates and variables in the input instance. 
\end{theorem}

Our algorithm for general MMSA$_t$ applies a recursion on the depth $t$, with our algorithms for Red-Blue Set Cover and MMSA$_4$ as the basis of the recursion. Each recursive step relies on an initially naive LP relaxation to which we add constraints as calls to the algorithm for smaller depth MMSA reveal new violated constraints.

We complement our upper bound for Red-Blue Set Cover with a hardness of approximation result.
\begin{theorem}\label{thm:hardness}
Assuming the Hypergraph Dense-vs-Random Conjecture, for every $\varepsilon > 0$, no polynomial-time algorithm achieves better than $O(m^{1/4-\varepsilon}/\log^2 k)$ approximation for Red-Blue Set Cover where $m$ is the number of sets and $k$ is the number of blue elements.    
\end{theorem}
To show the hardness, we present a reduction from Min $k$-Union to Red-Blue Set Cover that preserves the approximation up to a factor of $\polylog(k)$. Then, the hardness follows from the standard conjectured hardness of Min $k$-Union~\citep{chlamtavc2017minimizing}. In our reduction, all elements of the given instance of Min $k$-Union are considered as the red elements in the constructed instance for Red-Blue Set Cover and we further complement each set with a sample size of $O(\log k)$ (with replacement) from a ground set of blue elements of size $k$. We prove that this reduction is approximation-preserving up to a factor of $\polylog(k)$.

\paragraph{Organization.}
In Section~\ref{sec:prelim}, we restate Red-Blue Set Cover and introduce some notation.
In Section~\ref{sec:alg-red-blue}, we present our algorithm for Red-Blue Set Cover.
We adapt this algorithm for Partial Red-Blue Set Cover in Appendix~\ref{sec:partial-rbsc}.
Then, in Section~\ref{sec:alg-MMSA-t} we give the algorithm for MMSA$_t$ with $t\geq 5$. This algorithm relies on the algorithm for MMSA$_4$, which we describe later in Section~\ref{sec:alg-MMSA-4}.
We present a reduction from Min $k$-Union to Red-Blue Set Cover, which yields a hardness of approximation result for Red-Blue Set Cover, in Section~\ref{sec:reduction-k-union}. We discuss the hardness of the general MMSA$_{t}$ problem in Section~\ref{sec:lb-mmsa}. There we present a proof sketch that the integrality gap of the basic SDP relaxation strengthened by $\tilde\Omega(\log n)$ rounds of Sherali--Adams is $\Omega(N^{1-O((\log t)/t)})$ (where $N$ is the total number of gates and variables).

\section{Preliminaries}\label{sec:prelim}
To simplify the description and analysis of our approximation algorithm for Red-Blue Set Cover, we restate the problem in graph-theoretic terms. Essentially we restate the problem as MMSA$_3$. Specifically, we think of a Red-Blue Set Cover instance as a tripartite graph $(B,J,R,E)$ in which all edges ($E$) are incident on $J$ and either $B$ or $R$. The vertices in $J$ represent the set indices, and their neighbors in $B$ (resp.\ $R$) represent the blue (resp.\ red) elements in these sets. Thus, our goal is to find a subset of the vertices in $J$ that is a dominating set for $B$ and has a minimum total number of neighbors in $R$. For short, we will denote the cardinality of these different vertex sets by $k=|B|$, $m=|J|$, and $n=|R|$.

Similarly, we think of a MMSA$_4$ instance as a tuple $(B, J, R, S, E)$. Here, $B$, $J$, and $R$ represent gates in the second, third, and fourth layers of the circuit (where layer $i$ consists of the gates at distance $i-1$ from the root), respectively; $S$ represents the variables; $E$ represent edges between gates/variables. Combinatorially, the goal is to obtain a subset of $J$ as above, along with a minimum dominating set in $S$ for the red neighbors (in $R$) of our chosen subset of $J$.

%We also define a more general version of this problem that we call \emph{Partial Red-Blue Set Cover}. Here the input also includes a parameter $\hat k\in[k]$, and the goal is not to cover \emph{all} the blue elements, but at least $\hat k$ blue elements. In fact, our algorithm also gives the same approximation guarantee for Partial Red-Blue Set Cover, but that is not our focus here. We only define this generalization as a tool to approximate (total) Red-Blue Set Cover.

\paragraph{Notation.} We use $\Gamma(\cdot)$ to represent neighborhoods of vertices, and for a vertex set $U$, we use $\Gamma(U)$ to denote the union of neighborhoods of vertices in $U$, that is $\bigcup_{u\in U}\Gamma(u)$. We also consider restricted neighborhoods, which we denote by $\Gamma_T(u):=\Gamma(u)\cap T$ or $\Gamma_T(U)=\Gamma(U)\cap T$. We will refer to the cardinality of such a set, i.e.\ $|\Gamma_T(u)|$ as the \emph{$T$-degree} of $u$.

\begin{remark}
Note that, for every set index $j\in J$, the set $\Gamma(j)$ is simply the set $S_j$ in the set system formulation of the problem, and the set $\Gamma_R(j)$ (resp.\ $\Gamma_B(j)$) is simply the subset of red (resp.\ blue) elements in the set with index $j$. %is represented by $\Gamma_R(j)$ (resp.\ $\Gamma_B(j)$).$\Gamma_B(j) = B_j$ is simply the subset of blue elements in the set with index $i$, $\Gamma_R(j)= R_j$ is the subset of red elements in set $i$, and 
Similarly, $\Gamma_J(i)$ consists of indices $j$ representing those sets $S_j$ that contain element $i$, for any $i\in R\cup B$.
\end{remark}

For Red-Blue Set Cover algorithms, we introduce a natural notion of progress:%, as in {(\color{red} cite some papers?)}.

\begin{definition}
    We say that an algorithm for %(Partial)
    Red-Blue Set Cover \emph{makes progress towards an $O(\tilde A\cdot\log k)$-approximation} if, given an instance %\footnote{in an instance of Red-Blue Set Cover we would simply have $\hat k=k=|B|$} $\langle(B,J,R,E),\hat k\rangle$
    with an optimum solution containing $\opt$ red elements, the algorithm finds a subset $\hat J\subseteq J$ such that $$\frac{|\Gamma_R(\hat J)|}{|\Gamma_B(\hat J)|}\leq \tilde A\cdot\frac{\opt}{|B|}.$$ \label{def:progress}
\end{definition}

It is easy to see that if we have an algorithm which makes progress towards an $\tilde A$-approximation, then we can run this algorithm repeatedly (with decreasing $|B|$ parameter, where initially $|B|=k$) until we cover all blue elements and obtain an $O(\tilde A\cdot\log k)$-approximation.

For brevity, all logarithms are implicitly base 2 unless otherwise specified, that is, we write $\log(\cdot)\equiv\log_2(\cdot)$.

\section{Approximation Algorithm for Red-Blue Set Cover} \label{sec:alg-red-blue}
We begin by excluding a small number of red elements, and binning the sets $J$ into a small number of bins with uniform red-degree. For an $O(A)$-approximation, 
the goal will be to exclude at most $A\cdot\opt$ red elements (we may guess the value of $\opt$ by a simple linear or binary search). This is handled by the following lemma:
\begin{lemma}
    There is a polynomial time algorithm, which, given an input $(B,J,R,E)$ and parameter $n_0$, returns a set of at most $\log n$ pairs $(J_\alpha,R_\alpha)$ with the following properties:
    \begin{itemize}
        \item The sets $J_\alpha\subseteq J$ partition the set $J$.
        \item The sets $R_\alpha\subseteq R$ form a nested collection of sets, and the smallest among them (i.e., their intersection) has %all have
        cardinality at least $n-n_0$. That is, at most $n_0$ red elements are excluded by any of these sets.
        \item For every $\alpha$ there is some $r_\alpha$ such that every set $j\in J_\alpha$ has $R_\alpha$-degree (or restricted red set size)
        $$|\Gamma_{R_\alpha}(j)|\in[r_\alpha,2r_\alpha],$$
        \item and for every $\alpha$, every red element $i\in R_{\alpha}$ has $J_\alpha$-degree at most 
        (that is, the number of red sets in $J_\alpha$ that contain $i$)
        $$|\Gamma_{J_\alpha}(i)|\leq \frac{2mr_\alpha\log n}{n_0}.$$
    \end{itemize}
    \label{lem:red-size-binning}
\end{lemma}
\begin{proof}
    Consider the following algorithm:
    \begin{itemize}
        \item Let $r$ be the maximum red-degree (i.e., $\max_{j\in J}\deg_R(j)$).
        \item Repeat the following while $J\neq\emptyset$:
        \begin{itemize}
            \item Delete the top $n_0/\log n$ $J$-degree red elements from $R$, along with their incident edges.
            \item After this deletion, let 
            $J'=\{j\in J\mid \deg_R(j)\in[r/2,r]\}$.
            \item If $J'$ is non-empty, add the current pair $(J',R)$ to the list of output pairs (excluding all elements deleted from $R$ so far).
            \item Remove the sets in $J'$ from $J$ (along with incident edges) and let $r\leftarrow r/2$.
        \end{itemize}
    \end{itemize}
    By the decrease in $r$, it is easy to see that this partitions $J$ into at most $\log n$ sets (or more specifically, $\log$ of the initial maximum red set size, $\max_{j\in J}\deg_R(j)$). Also note that at the beginning of each iteration, all red sets have size at most $r$, and so there are at most $mr$ edges to $R$, and the top $n_0/\log n$ $J$-degree red elements will have average degree (and in particular minimum degree) at most $mr/(n_0/\log n)$. Thus, after removing these red elements, all remaining red elements will have $J$-degree (and in particular $J'$-degree) at most the required bound of $\frac{2 m r_\alpha \log n}{n_0}$ where $r_\alpha = r/2$. Finally, since there at most $\log n$ iterations, the total number of red elements removed across all iterations is at most $n_0$.
\end{proof}

Our algorithm works in iterations, where at every iteration, some subset of blue elements is covered and removed from $B$. However, nothing is removed from $J$ or $R$. Thus the above lemma applies throughout the algorithm. Note that for an optimum solution $J_\opt\subseteq J$, for at least one of the subsets $J_\alpha$ in the above partition, the sets in $J_\opt\cap J_\alpha$ must cover at least a $(1/\log n)$-fraction of $B$. Thus, to achieve an $O(A)$ approximation, it suffices to apply the above lemma with parameter %Since we can repeatedly find a partition as above (for 
$n_0=\opt\cdot A$, and repeatedly make progress towards an $A$-approximation within one of the subgraphs induced on $(B,J_\alpha,R_\alpha)$. We will only pay at most another $\opt\cdot A$ in the final analysis by restricting our attention to these subinstances. %which excludes at most and run an approximation algorithm just on the relevant subinstance, we get the following corollary:
\iffalse
\begin{corollary}
    Given an $A$-approximation for Partial Red-Blue Set Cover on instances where, for some fixed $r$, every red set has size in the range $[r,2r]$ and %all red
    each red element 
    is contained in at most $2mr\log n/(A\cdot\opt)$ sets, then there is also an $O(A\cdot\log n
    \log k)$-approximation for general Red-Blue Set Cover.
    \label{cor:main}
\end{corollary}
\fi

Let us fix some optimum solution $J_\opt\subseteq J$ in advance. For any $\alpha$ in the above partition, let $J^{\alpha}_{\opt}=J_{\alpha}\cap J_{\opt}$ be the collection of sets in $J_\alpha$ that also belong to our optimum solution, and let $B_\alpha=\Gamma_B(J^{\alpha}_{\opt})$ be the blue elements covered by the sets in $J^{\alpha}_{\opt}$. Note that every blue element must belong by the feasibility of $J_\opt$ to at least one $B_\alpha$. %Thus, let us assume that we have an instance $(B,J',R',E)$ such as above, where every $j\in J'$ has $R'$-degree $\deg_{R'}(j)\in[r,2r]$, and every red element $i\in R'$ has $J'$-degree at most $\deg_{J'}(i)\leq 2mr\log n/(\opt\cdot A)$. In this instance, let $\hat J_{\opt}\subseteq J'$ be a solution with at most $\opt$ red neighbors that covers a blue set $\hat B\subseteq B$ with initial cardinality (before various iterations that make progress towards our approximation) $|\hat B|=\hat k\geq k/\log n$.\footnote{For brevity of notation, we always denote the instance as $(B,J,R,E)$, though the set $B$ will shrink by at least a $(1-1/\log)$ factor after every outer iteration where we cover $\hat B$ and remove these elements from $B$.} 
In this context we can show the following:

\begin{lemma}
    For any $\alpha$ in the partition described in Lemma~\ref{lem:red-size-binning}, there exists a red element $i_0\in R_{\alpha}$ such that its optimum $J_{\alpha}$-restricted neighbors $\Gamma_{J^{\alpha}_{\opt}}(i_0)$ cover at least $|B_\alpha|r_\alpha/\opt$ blue elements. %in $B_{\alpha}$.
    \label{lem:back-deg}
\end{lemma}
\begin{proof}
    Consider the following subgraph of $(B_\alpha,J^\alpha_\opt,E(B_\alpha, J^\alpha_{\opt}))$. For every blue element $\ell\in B_\alpha$, retain exactly one edge to $J^\alpha_{\opt}$. Let $\hat F$ be this set of edges.
    
    Thus the blue elements $B_{\alpha}$ have at least $|B_{\alpha}|r_\alpha$ paths through $\hat F\times E(J^\alpha_{\opt},R_{\alpha})$ to the red neighbors of $J^\alpha_{\opt}$ in $R_\alpha$. Since there are at most $\opt$ such red neighbors, at least one of them, say $i_0\in\Gamma_{R_{\alpha}}(J^\alpha_{\opt})$, must be involved in at least a $1/\opt$ fraction of these paths. That is, at least $|B_{\alpha}|r_{\alpha}/\opt$ paths. Since the $\hat F$-neighborhoods of the vertices in $J^\alpha_{\opt}$ are disjoint (by construction), these paths end in distinct blue elements, thus, at least $|B_{\alpha}|r_{\alpha}/\opt$ elements in $B_{\alpha}$.
\end{proof}

Of course, we cannot know which red element will have this property, but the algorithm can try all elements and run the remainder of the algorithm on each one. Now, given a red element $i_0\in R_\alpha$, our algorithm proceeds as follows: Begin by solving the following LP.

\begin{align}
\text{max\ }&\sum_{\ell\in B}z_\ell \label{LP:blue}\\
&\sum_{i\in R_\alpha}y_i\leq \opt\label{LP:red}\\
&0\leq z_\ell\leq\min\{1,\sum_{j\in \Gamma_{J_\alpha}(i_0)\cap\Gamma_{J_\alpha}(\ell)} x_j\}&\forall \ell\in B\\
& 0\leq x_j\leq y_i\leq 1 &\forall j\in \Gamma_{J_\alpha}(i_0),i\in \Gamma_{R_\alpha}(j)\label{LP:red-expansion}
%&\sum_{\ell\in B}z_\ell\geq \frac{|\hat B|r}{\opt}&\label{LP:blue}
\end{align}

%If there is not feasible solution to this LP, abort. Note however that Lemma~\ref{lem:back-deg} guarantees that for at least one $i_0\in R'$ there will be a feasible integral solution: 
In the intended solution, $x_j$ is the indicator for the event that $j\in \Gamma_{J^\alpha_{\opt}}(i_0)$, 
$y_i$ is the indicator variable for the event that red vertex $i$ is in the union of red sets indexed by $\Gamma_{J^\alpha_\opt}(i_0)$ (and therefore in the optimum solution), and $z_\ell$ is the indicator variable for the event that the blue vertex $\ell\in B$ is is covered by some set in $\Gamma_{J^\alpha_\opt}(i_0)$. %(assuming we have guessed a correct value of $\opt$ and chosen the correct subinstance in the partition).
This LP is always feasible (say, by setting all variables to $0$), though since there are at most $\log n$ subinstances in the partition, for at least one $\alpha$ we must have $|B_\alpha|\geq |B|/\log n$, and then by Lemma~\ref{lem:back-deg}, there is also some choice of $i_0\in R_\alpha$ for which the objective function satisfies

\begin{equation}
    \sum_{\ell\in B}z_\ell\geq \frac{|B_\alpha| r_\alpha}{\opt}\geq \frac{|B|r_\alpha}{\opt\cdot\log n}.\label{eq:blue-LP-bound}
\end{equation}

The algorithm will choose $\alpha$ and $i_0$ that maximize the rescaled objective function $\sum_{\ell\in B}z_\ell/r_\alpha$, guaranteeing this bound. Finally, at this point, we perform a simple randomized rounding, choosing every set $j\in \Gamma_{J_\alpha}(i_0)$ independently with probability $x_j$. The entire algorithm is described in Algorithm~\ref{alg-main}.

\begin{algorithm}[t]
    \caption{Approximation Algorithm for Blue-Red Set Cover}
    \label{alg-main}
\begin{algorithmic}
\State \textbf{Input:} $B, J, R, E$
    \State{\textbf{guess} $\opt$} \Comment{e.g. using binary search}
    \State{$J_\alg = \emptyset$}\Comment{$J_{\alg}$ stores the current solution}
    \State{find} decomposition $\{(J_\alpha, R_\alpha)\}_\alpha$ as in Lemma~\ref{lem:red-size-binning}, with parameter $n_0=\opt\cdot m^{1/3}\log^{4/3}n\log^2 k.$ %$$n_0=\opt\cdot m^{1/3}\log^{4/3}n\log^2 k.$$
    \While{$B\neq \emptyset$}
    \For{all $\alpha$ and $i_0\in R_\alpha$}
       \State{Solve LP (\ref{LP:blue})-(\ref{LP:red-expansion})}. Let $\mathrm{LP}(\alpha, i_0)$ be its objective value
       %with $R' = R_\alpha$, $J' = J_{\alpha}$}
    \EndFor
    \State{\textbf{choose} the value of $\alpha$ and $i_0$ which maximizes $\mathrm{LP}(\alpha, i_0) / r_\alpha$} 
    %the objective function of the LP}
    \State{\textbf{let} $x, y, z$ be an optimal solution for this LP}
    \State{use solution $x$ and the method of conditional expectations to find $J^*\subseteq J_\alpha$ 
    \Statex\quad\quad s.t. 
    $|\Gamma_{R_\alpha}(J^*)|/|\Gamma_B(J^*)| \leq O(1)\cdot m^{1/3}\log^{4/3}n \cdot \opt/|B|$}\Comment{see the proof for details}
\State{\textbf{let} $J_\alg = J_\alg \cup J^*$}
\State{\textbf{let} $B = B\setminus \Gamma_B(J^*)$}
\State{\textbf{remove} edges incident to deleted vertices from $E$} 
\EndWhile
\State{\textbf{return} $J_\alg$}
\end{algorithmic}
\end{algorithm}

Now let $J^*\subseteq J_\alpha$ be the collection of sets added by this step in the algorithm. Let us analyze the number of blue elements covered by $J^*$ and the number of red elements added to the solution. First, noting that this LP acts as a max coverage relaxation for blue elements, the expected number of blue elements covered will be at least $$(1-1/e)\frac{|B|r_\alpha}{\opt\cdot\log n},$$
by the standard analysis and the bound~\eqref{eq:blue-LP-bound}.

Now let us bound the number of red elements added. Let $$J_+=\left\{j\in \Gamma_{J_\alpha}(i_0)\;\left|\; x_j\geq\frac{\opt}{r_{\alpha} \hat A}\right.\right\}$$ for a value of $\hat A$ to be determined later. By Constraint~\eqref{LP:red-expansion}, every red neighbor $i\in \Gamma_{R_\alpha}(J_+)$ will also have $y_i\geq \opt/(r_{\alpha}\hat A)$, and so by Constraint~\eqref{LP:red}, there can be at most $r_{\alpha}\hat A$ such neighbors. On the other hand, the expected number of red elements added by the remaining sets $j\in J^*\setminus J_+$ is bounded by
\begin{align*}
    \expec\left[\left|\bigcup_{j\in J^*\setminus J_+}\Gamma_{R_{\alpha}}(j)\right|\right] &\leq 2r_{\alpha}\expec\left[\left|J^*\setminus J_+\right|\right]&\text{by $R_{\alpha}$-degree bounds for $j\in J_{\alpha}$}\\
    &= 2r_{\alpha}\cdot\sum_{j\in\Gamma_{J_{\alpha}}(i_0)\setminus J_+}x_j\\
    &\leq 2r_{\alpha}\cdot\frac{\opt}{r_\alpha \hat A} \left|\Gamma_{J_{\alpha}}(i_0)\setminus J_+\right|&\text{by definition of }J_+\\
    &\leq \frac{2\opt}{\hat A}\cdot\frac{2mr_{\alpha}\log n}{\opt\cdot \hat A}&\text{by $J_{\alpha}$-degree bounds for $i\in R_{\alpha}$}\\
    &=\frac{4mr_{\alpha}\log n}{\hat A^2}
\end{align*}
These two bounds are equal when $r_{\alpha}\hat A=4mr_{\alpha}\log n/\hat A^2$, that is, when $\hat A=(4m\log n)^{1/3}$, giving us a total bound on the expected number of red elements added in this step of
$$\expec[|\Gamma_{R_{\alpha}}(J^*)|]\leq 2r_{\alpha}(4m\log n)^{1/3}\leq \frac{2r_{\alpha}(4m\log n)^{1/3}}{(1-1/e)|B|r_\alpha/(\opt\cdot\log n)}\cdot \expec[|\Gamma_B(J^*)|].$$
Thus, 
$$\expec[|\Gamma_{R_{\alpha}}(J^*)|] - 
\frac{2(4m\log^4 n)^{1/3}\opt}{(1-1/e)|B|} \cdot \expec{|\Gamma_B(J^*)|}
\leq 0.$$
%Now,
%$$\frac{2r(4m\log n)^{1/3}}{(1-1/e)|\hat B|r_{\alpha}/\opt} = O\left(m^{1/3}\log^{1/3}n\right).$$
%Therefore, 
%$\expec[|\Gamma_{R_{\alpha}}(J^*)|] - 
%c m^{1/3}\log^{1/3}n
%\expec[|\Gamma_B(J^*)|]
%\leq 0$ for an appropriately chosen constant $c >0$.
Using the method of conditional expectations, we can derandomize the algorithm and find $J^*$ with a non-empty blue neighbor set such that $$\frac{|\Gamma_{R_{\alpha}}(J^*)|}{|\Gamma_B(J^*)|} \leq 
O(1)\cdot m^{1/3}\log^{4/3}n\cdot\frac{\opt}{|B|}.$$
Thus, we make progress (according to Definition~\ref{def:progress})
towards an approximation guarantee of $\tilde A\cdot k$ for $\tilde A = O\left(m^{1/3}\log^{4/3}n\right)$,
which, as noted, ultimately gives us the same approximation guarantee for Red-Blue Set Cover, proving Theorem~\ref{thm:main}. %Finally, using Corollary~\ref{cor:main} we get an $O(m^{1/3}\log^{4/3}n\log^2 k)$-approximation algorithm for Red-Blue Set Cover.

\section{Approximating MMSA\texorpdfstring{$_{t}$}{t} for \texorpdfstring{$t\geq 5$}{t >= 5}}
\label{sec:alg-MMSA-t}
We now turn to the general problem of approximating MMSA$_t$ for arbitrarily large (but fixed) $t$. We will build on our approximation algorithm for %MMSA$_3$ (i.e., Red-Blue Set Cover) and
MMSA$_4$ (described in Section~\ref{sec:alg-MMSA-4}) by recursively calling approximation algorithms for the problem with smaller values of $t$, and using the result of this approximation as a separation oracle in certain cases.

We will denote the total size of our input by $N$, and we will denote our approximation factor for MMSA$_t$ by $A_t$. We will only describe an algorithm for even depth. There is a very slightly simpler but quite similar algorithm for odd depth, however the guarantee we are able to achieve for MMSA$_{2t-1}$ is nearly the same as for MMSA$_{2t}$ (up to an $O(\log N)$ factor). Since MMSA$_{2t-1}$ is essentially a special case of MMSA$_{2t}$, we focus only on even levels. %We will begin by presenting our guarantee for odd $t$, since it is very slightly simpler, though the algorithm for even $t$ is quite similar.

\begin{lemma}
  For $t\geq 2$, if MMSA$_{2t}$ can be approximated to within a factor of $A_{2t}$, then we can approximate MMSA$_{2t+2}$ (and thus MMSA$_{2t+1}$) to within $O(\sqrt{N\cdot A_{2t}}\log N)$. 
  \label{lem:mmsa-even}
\end{lemma}
\begin{proof}
  Denote our input as a layered graph with vertex layers $V_1,\ldots,V_{2t+2}$. Ideally, we would like to discard any vertex $j\in V_{2t}$ such that covering its neighbors $\Gamma_{2t+1}(j)$ requires more than $\opt$ vertices in $V_{2t+2}$, however, checking this precisely requires solving Set Cover. Instead, we discard any vertex $j\in V_{2t}$ for which the smallest {\em fractional} set cover\footnote{That is, $\min\bigg\{{\displaystyle\sum_{h\in S}z_h}\;\bigg|\; \forall i\in \Gamma_{V_{2t+1}}(j):{\displaystyle\sum_{h\in\Gamma_{V_{2t+2}}(i)}}z_h\geq 1;\;\forall h\in S:z_h\geq 0\bigg\}$.} in $V_{2t+2}$ of its neighbors $\Gamma_{V_{2t+1}}(j)$ has value greater than $\opt$. Such vertices cannot be included without incurring cost greater than $\opt$ and so we know they do not participate in any optimum solution. We begin with the following basic LP:
  \begin{align}
      &\sum_{h\in V_{2t+2}}w_h\leq\opt\label{LPt:opt}\\
      &y_i\leq \sum_{h\in\Gamma_{2t+2}(i)}w_h&\forall i\in V_{2t+1}\label{LPt:cover}\\
      &x_j\leq y_i &\forall j\in V_{2t},i\in\Gamma_{V_{2t+1}}(j)\label{LPt:neighbors}\\
      &x_j,y_i,w_h\in[0,1]&\forall j\in V_{2t}\forall i\in V_{2t+1}\forall h\in V_{2t+2}
  \end{align}
  Note that, as stated, this LP is trivial. Indeed, in the absence of any additional constraints, the all-zero solution is feasible. However, we will add new violated constraints as the algorithm proceeds. 
  
  Given a solution to the above linear program, our algorithm for MMSA$_{2t+2}$ is as follows:
  \begin{itemize}
      \item Let $V_{2t}^+=\{j\in V_{2t}\mid x_j\geq 2(1+\ln N)/A_{2t+2}\}$. Add these vertices to the solution.
      \item Let $V_{2t+2}^+=\Gamma_{V_{2t+2}}(\Gamma_{V_{2t+1}}(V_{2t}^+))$ be the neighbors-of-neighbors of $V_{2t}^+$.
      %\item For every vertex $h\in V_{2t+1}$, let $w^+(h)=\min\{1,w_h\cdot A_{2t+2}/(2(1+\ln N))\}$.
      \item Apply a greedy $(1+\ln N)$-approximation for Set Cover to obtain a set cover (in $V_{2t+2}$) for $\Gamma_{V_{2t+1}}(V_{2t}^+)$, and add this set cover to the solution as well.
      \item Create an instance of MMSA$_{2t}$ by removing layers $V_{2t+1}$, $V_{2t+2}$, all vertices in $V_{2t}^+$, as well as their neighbors in $V_{2t-1}$, that is, $\Gamma_{V_{2t-1}}(V_{2t}^+)$, as these are already covered by $V_{2t}^+$.
      \item Apply an $A_{2t}$-approximation algorithm for MMSA$_{2t}$ to this instance, and let $U_{\alg}\subseteq V_{2t}\setminus V_{2t}^+$ be the result (or at least the portion belonging to layer $2t$).
      \item If $|U_{\alg}|\leq A_{2t+2}/(2+2\ln N)$, add the vertices in $U_{\alg}$ to the solution, as well as a greedy set cover (in $V_{2t+2}$) for the neighborhood $\Gamma_{V_{2t+1}}(U_{\alg})$.
      \item Otherwise (if $|U_{\alg}|> A_{2t+2}/(2+2\ln N)$), continue the Ellipsoid algorithm using the new violated constraint
      \begin{equation}
          \sum_{j\in V_{2t}\setminus V_{2t}^+}x_j\geq \left\lfloor\frac{A_{2t+2}}{2(1+\ln N) A_{2t}}\right\rfloor+1,
          \label{LPt:violated}
      \end{equation}
      and restart the algorithm (discarding the previous solution) using the new LP solution.
  \end{itemize}
  
  Let us now analyze this algorithm. By~\eqref{LPt:neighbors}, we know that all neighbors $i\in V_{2t+1}$ of $V_{2t}^+$ have LP value $y_i\geq 2(1+\ln N)/A_{2t+2}$. Thus, by~\eqref{LPt:cover}, if for every vertex $h\in V_{2t+2}$ we define $w^+_h=w_h\cdot A_{2t+2}/(2+2\ln N)$, then this is a fractional Set Cover for the $V_{2t+1}$-neighborhood $\Gamma_{V_{2t+1}}(V_{2t}^+)$, and by~\eqref{LPt:opt} it has total fractional value at most $\opt\cdot A_{2t+2}/(2+2\ln N)$. Thus, the greedy Set Cover $(1+\ln N)$-approximation algorithm will cover this neighborhood using at most $\opt\cdot A_{2t+2}/2$ vertices in $V_{2t+2}$. After this step, we may add at most $\opt\cdot A_{2t+2}/2$ additional vertices in $V_{2t+2}$ to our solution to obtain an $A_{2t+2}$-approximation.
  
  Now, suppose our MMSA$_{2t}$ approximation returns a set $U_{\alg}$ of cardinality $|U_{\alg}|\leq A_{2t+2}/(2+2\ln N)$. Clearly, adding to our solution the vertices of $U_{\alg}$ and a $V_{2t+2}$-Set Cover for its neighborhood $\Gamma_{V_{2t+1}}(U_{\alg})$ gives a feasible solution to our MMSA$_{2t+2}$ instance. Moreover, since by our preprocessing, the neighborhood $\Gamma_{V_{2t+1}}(j)$ of every $j\in U_{\alg}$ has a fractional Set Cover in $V_{2t+2}$ of value at most $\opt$, it follows that the union of all these neighborhoods, that is $\Gamma_{V_{2t+1}}(U_{\alg})$, has a fractional set cover in $V_{2t+2}$ of value at most $\opt\cdot|U_{\alg}|\leq\opt\cdot A_{2t+2}/(2+2\ln N)$. And so applying a greedy Set Cover algorithm for the neighborhood $\Gamma_{V_{2t+1}}(U_{\alg})$ contributes at most an additional $\opt\cdot A_{2t+2}/2$ vertices in $V_{2t+2}$ to our solution, as required.
  
  Finally, let us examine the validity of the final step (the separation oracle). If the $A_{2t}$-approximation for MMSA$_{2t}$ was not able to find a solution of size at most $A_{2t+2}/(2+2\ln N)$, then by definition, the value of any solution to our MMSA$_{2t}$ instance must be greater than $A_{2t+2}/((2+2\ln N)A_{2t})$. This is a subinstance of our original instance, so any solution to our original MMSA$_{2t+2}$ instance must also contain more than $A_{2t+2}/((2+2\ln N)A_{2t})$ vertices in $V_{2t}$. Thus, Constraint~\eqref{LPt:violated} is valid for any optimum solution. But when is it violated?
  
  By definition of $V_{2t}^+$, the current total LP value of $V_{2t}\setminus V_{2t}^+$ is at most $2(1+\ln N)N/A_{2t+2}$. And so the current LP solution violates~\eqref{LPt:violated} if
  $$\frac{2(1+\ln N)N}{A_{2t+2}} \leq\frac{A_{2t+2}}{2(1+\ln N)A_{2t}}\qquad\Longleftrightarrow\qquad A_{2t+2}\geq2(1+\ln N)\sqrt{N\cdot A_{2t}}.$$ Thus, we can obtain an approximation guarantee of $A_{2t+2}=O(\sqrt{N\cdot A_{2t}}\log N)$ as claimed.
\end{proof}

Thus, by induction on $t$, with the guarantee %s of Theorem~\ref{thm:main} for MMSA$_3$ (Red-Blue Set Cover) and 
of Theorem~\ref{thm:main-MMSA4} for MMSA$_4$ as the basis of the induction, and Lemma%~\ref{lem:mmsa-odd} and
~\ref{lem:mmsa-even} for the inductive steps, we get a general approximation algorithm for MMSA$_t$ with approximation ratio $$O\left(N^{1-\frac132^{3-\lceil t/2\rceil}}\cdot(\log N)^{2+O(2^{-t/2})}\right).$$
%where once again we note that the induction gives the result for even $t$ explicitly, with the approximation for (smaller) odd $t$ following as a special case.

\section{Reduction from Min \texorpdfstring{$k$}{k}-Union to Red-Blue Set Cover}\label{sec:reduction-k-union}
In this section, we first present a reduction from Min $k$-Union to Red-Blue Set Cover and then prove a hardness result for Red-Blue Set Cover. We start with formally defining the Min $k$-Union problem.

\begin{definition}[Min $k$-Union] 
In the Min $k$-Union problem, we are given a set $X$ of size $n$, a family $\cal S$ of $m$ sets $S_1,\dots, S_m$, and an integer parameter $k\geq 1$. The goal is to choose $k$ sets $S_{i_1},\dots,S_{i_k}$ so as to minimize the cost $\left|\bigcup_{t=1}^k S_{i_t}\right|$.
We will denote the cost of the optimal solution by $\mathrm{OPT}_{MU}(X, {\cal S}, k)$.
\end{definition}
Note that Min $k$-Union resembles the Red-Blue Set Cover Cover problem: in both problems, the goal is to choose some subsets $S_{i_1},\dots, S_{i_r}$ from a given family $\cal S$ so as to minimize the number of elements or red elements in their union. Importantly, however, the feasibility requirements on the chosen subsets $S_{i_1},\dots, S_{i_r}$ are different in  Red-Blue Set Cover Cover and Min $k$-Union; in the former, we require that the chosen sets cover all $k$ blue points but in the latter, we simply require that the number of chosen sets be $k$.
Despite this difference, we show that there is a simple reduction from Min $k$-Union to Red-Blue Set Cover.
\begin{claim}
There is a randomized polynomial-time reduction from Min $k$-Union to Red-Blue Set Cover that given an instance ${\cal I} = (X, {\cal S}, k)$ of Min $k$-Union returns an instance ${\cal I}' = (R, B, \{S_i'\}_{i\in[m]})$ of Red-Blue Set Cover satisfying the following two properties:
\begin{enumerate}
\item If $S'_{j_1},\dots,S'_{j_r}$ is a feasible solution for ${\cal I}'$ then $k'\leq r$ and the cost of solution $S_{j_1},\dots,S_{j_{k'}}$ for Min $k'$-Union where $k' = \lfloor k/\ell \rfloor$ and $\ell=\lceil \log_e k\rceil + 1$ does not exceed that of $S'_{j_1},\dots,S'_{j_r}$ for Red-Blue Set Cover: 
$$\mathrm{cost}_{MU}(S_{j_1},\dots,S_{j_{k'}})\equiv \left|\bigcup_{t=1}^{k'} S_{i_t}\right| \leq \left|\bigcup_{t=1}^{r} (S'_{i_t}\cap R)\right| = \mathrm{cost}_{RB}(S'_{j_1},\dots,S'_{j_{r}}).$$ 
This is true always no matter what random choices the reduction makes.
\item  $\mathrm{OPT}_{RB}(R,B, \{S'_i\}_i) \leq \mathrm{OPT}_{MU}(X, {\cal S}, k)$ with probability at least $1 - 1/e$.
\end{enumerate}
\end{claim}
\begin{proof}
We define instance $\cal I'$ as follows. Let $R = X$ and $B = [k]$. For every $i\in [m]$, let $R_i = S_i$; $B_i$ be a set of $\ell$ elements randomly sampled from $[k]$ with replacement, and $S_i' = R_i \cup B_i$. All random choices are independent. Now we verify that this reduction satisfies the required properties.

Consider a feasible solution $S_{j_1}',\dots,S_{j_r}'$ for ${\cal I}'$. Since this solution is feasible, $\cup_{t=1}^r B_t = B$. Now $|B_t| \leq \ell$ and thus $r \geq |B|/\ell = k/\ell \geq k'$, as required. Further,
$$\mathrm{cost}_{MU}(S_{j_1},\dots,S_{j_{k'}})\equiv \left|\bigcup_{t=1}^{k'} S_{j_t}\right| = \left|\bigcup_{t=1}^{k'} R_{j_t}\right| \leq \left|\bigcup_{t=1}^{r} R_{j_t}\right| \equiv \mathrm{cost}_{RB}(S'_{j_1},\dots,S'_{j_{r}}).$$
We have verified that item 1 holds.
Now, let $S_{i_1},\dots,S_{i_k}$ be an optimal solution for $\cal I$.
We claim that $S_{i_1}',\dots,S_{i_k}'$ is a feasible solution for ${\cal I}'$ with probability at least $1- 1/e$. To verify this claim, we need to lower bound the probability that $B_{i_1}\cup \dots\cup B_{i_k} = B$. Indeed, set $B_{i_1}\cup \dots\cup B_{i_k}$ consists of $k\ell$ elements sampled from $B$ with replacement. The probability that a given element $b\in B$ is not in 
$B_{i_1}\cup \dots\cup B_{i_k}$ is at most $(1-1/k)^{k\ell}\leq e^{-\ell} \leq \frac{1}{ek}$. By the union bound, the probability that there is some $b\in B \setminus (B_{i_1}\cup \dots\cup B_{i_k})$ is at most $k \times \frac{1}{ek} = \frac{1}{e}$. Thus, there is no such $b$ with probability at least $1 - 1/e$ and consequently $B_{i_1}\cup \dots\cup B_{i_k} = B$.  In that case, the cost of solution $S'_{i_1}, \dots, S'_{i_k}$ for Red-Blue Set Cover equals
$\left|\bigcup_{i=1}^k R'_{i_t}\right| = \left|\bigcup_{i=1}^k S_{i_t}\right|$, the cost of the optimal solution for Min $k$-Union.
\end{proof}

\begin{corollary} Assume that there is an $\alpha(m,n)$ approximation algorithm $\cal A$ for Red-Blue Set Cover (where $\alpha$ is a non-decreasing function of $m$ and $n$).Then there exists a randomized polynomial-time algorithm $\cal B$ for Min $k$-Union that finds $k'$ sets $S_{i_1},\dots, S_{i_{k'}}$ such that
$$\left| \bigcup_{t=1}^{k'} S_{i_t}\right|\leq \alpha(m,n) \mathrm{OPT}_{MU}(X, {\cal S}, k).$$ The failure probability is at most $1/n$.
\end{corollary}
\begin{proof}
We simply apply the reduction to the input instance of Min $k$-Union and then solve the obtained instance of Red-Blue Set Cover using algorithm $\cal A$. To make sure that the failure probability is at most $1/n$, we repeat this procedure $\lceil\log_e n\rceil$ times and output the best of the solutions we found.
\end{proof}
\begin{theorem}
Assume that there is an $\alpha(m,n)$ approximation algorithm $\cal A$ for Red-Blue Set Cover (where $\alpha$ is a non-decreasing function of $m$ and $n$).
Then there exists an $O(\log^2 k) \alpha(m,n)$ approximation algorithm for Min $k$-Union.
\end{theorem}
\begin{proof}
Our algorithm iteratively uses algorithm $\cal B$ from the corollary to find an approximate solution. First, it runs $\cal B$ on the input instance and gets $k_1 = k'$ sets. Then it removes the sets it found from the instance and reduces the parameter $k$ to $k -k_1$. Then the algorithm runs $\cal B$ on the obtained instance and gets $k_2$ sets. It again removes the obtained sets and reduces $k$ to $k - k_1 - k_2$ (here $k$ is the original value of $k$). It repeats these steps over and over until it finds $k$ sets in total. That is, $k_1 + \dots + k_T = k$ where $T$ is the number of iterations the algorithm performs. 

Observe that each of the instances of Min $k$-Union constructed in this process has cost at most $\mathrm{OPT}_{MU}(X, {\cal S}, k)$. Indeed, consider the subinstance ${\cal I}_{t+1}$ we solve at iteration $t+1$. Consider $k$ sets that form an optimal solution for $(X, {\cal S}, k)$. At most $k_1 +\dots + k_t$ of them have been removed from ${\cal I}_{t+1}$ and thus at least $k -k_1 -\dots - k_t$ are still present in ${\cal I}_{t+1}$. Let us arbitrarily choose $k -k_1 -\dots - k_t$  sets among them. The chosen sets form a feasible solution for ${\cal I}_{t+1}$ of cost at most 
$\mathrm{OPT}_{MU}(X, {\cal S}, k)$.

Thus, the cost of a partial solution we find at each iteration $t$ is at most $\alpha(m,n) \cdot\mathrm{OPT}_{MU}(X, {\cal S}, k)$. The total cost is at most $\alpha(m,n) \cdot T \cdot\mathrm{OPT}_{MU}(X, {\cal S}, k)$. It remains to show that $T\leq O(\log^2 k)$. We observe that the value of $k$ reduces by a factor at least $1 - 1/\ell$ in each iteration, thus after $t$ iterations it is at most $(1-1/\ell)^t k$. We conclude that the total number of iterations $T$ is at most $O(\ell \log k) = O(\log^2 k)$, as desired.
\end{proof}
Now we obtain a conditional hardness result for Reb-Blue Set Cover from a corollary from the Hypergraph Dense-vs-Random Conjecture.
\begin{corollary}[\cite{chlamtavc2017minimizing}]
Assuming the Hypergraph Dense-vs-Random Conjecture, for every $\varepsilon > 0$, no polynomial-time algorithm for Min $k$-Union achieves better than $\Omega(m^{1/4-\varepsilon})$ approximation. 
\end{corollary}

Theorem~\ref{thm:hardness} immediately follows.

\section{Approximation Algorithm for MMSA\texorpdfstring{$_{4}$}{4}}
\label{sec:alg-MMSA-4}
Consider an instance $(B,J,R,S,E)$ of MMSA$_4$. As we did for Red-Blue Set Cover, we will focus on making progress towards a good approximation.

\begin{definition}
    We say that an algorithm for %(Partial)
    MMSA$_4$ \emph{makes progress towards an $ O(A)$-approximation} if, given an instance %\footnote{in an instance of Red-Blue Set Cover we would simply have $\hat k=k=|B|$} $\langle(B,J,R,E),\hat k\rangle$
    with an optimum solution containing at most $\opt$ vertices in $S$, the algorithm finds a subset $\hat J\subseteq J$ and a subset $\hat S\subseteq S$ such that $\Gamma_R(\hat J)\subseteq\Gamma_R(\hat S)$ (a valid partial solution) and  $$\frac{|\hat S|}{|\Gamma_B(\hat J)|}\leq A\cdot\frac{\opt}{|B|}.$$ %\label{def:progress}
\end{definition}

As before, it is easy to see that given such an algorithm, we can run such an algorithm repeatedly to obtain an actual $\tilde O(A)$ approximation for MMSA$_4$. In fact, in the rest of this section we will {\em only} discuss an algorithm which makes progress towards an $O(A)$-approximation.

For the sake of formulating an LP relaxation with a high degree of uniformity, we will actually focus on a partial solution which covers a large fraction of blue elements in a uniform manner:

\begin{lemma}
  For any cover $J_0\subseteq J$ of the blue elements $B$, there exist subsets $J'\subseteq J_0$ and $B'\subseteq B$ and a parameter $\Delta>0$ with the following properties:
  \begin{itemize}
    \item Every vertex $j\in J'$ has $B'$-degree in the range $\deg_{B'}(j)\in[\Delta,2\Delta]$.
    \item Every blue element $\ell\in B'$ has at least one neighbor in $J'_{\Delta}$ and at most $2e\ln(2k)$ neighbors.
    \item We have the cardinality bound $|B'|\geq %\rho
    |B|/(\log k\log m)$.
\end{itemize}
  \label{lem:mmsa4-partial-solution}
\end{lemma}
\begin{proof}
  Partitioning the blue elements by their $J_0$-degrees, there is some $D$ such that the set $B_D=\{\ell\in B\mid \deg_{J_0}(\ell)\in[D,2D]\}$ has cardinality $|B_D|\geq |B|/\log m$. If we sampled every $j\in J_0$ independently with probability $\ln(2k)/D$, we would get some subset $J'\subseteq J_0$ and with positive probability, every $\ell\in B_D$ would still be covered by $J'$, and no $\ell\in B_D$ would have more than $2e\ln(2k)$ neighbors in $J'$. Finally, if we binned the vertices in $J'$ by their $B_D$-degrees, then some bin $J'_{\Delta}=\{j\in J'\mid \deg_{B_D}(j)\in[\Delta,2\Delta]\}$ would cover at least a $1/\log k$ fraction of elements in $B_D$, call them $B'$, and every element $\ell\in B'$ would still have (at least one and) at most $2e\ln(2k)$ neighbors in $J'_{\Delta}$.
\end{proof}

\paragraph{Simplifying assumptions.} We can make the following assumptions which will be useful in the analysis of our algorithm. First, we may assume that for every $j\in J$, the red neighborhood $\Gamma_R(j)$ has a fractional set cover in $S$ of weight at most $\opt$. That is, the standard LP relaxation for covering $\Gamma_R(j)$ using $S$ has optimum value at most $\opt$. If we have guessed the correct value of $\opt$, then we know that no $j\in J$ whose red neighborhood cannot be covered with cost $\opt$ can participate in an optimum solution, and can therefore be discarded. We may also assume that for some $\eps>0$, the value $\Delta$ above is at most $k/m^{\eps}$. The reason is that otherwise, the blue elements $B'$ can be covered with at most $\tilde O(m^{\eps})$ vertices in $J$, and we know that for each of these, its red neighborhood can be covered by a set of size $\tilde O(\opt)$ in $S$, and thus we can make progress towards an $\tilde O(m^{\eps})$ approximation.

Guessing the value of $\Delta\in[k]$ above and the value of the optimum $\opt$, we can write the following LP relaxation:
\begin{align}
&\sum_{h\in S}w_h\leq \opt\\
&\sum_{\ell\in B}z_{\ell}\geq |B|/(\log k\log m)\label{LP4:B-weight}\\
&z_{\ell}\leq \sum_{j\in\Gamma_{J}(\ell)}x_j^{\ell}\leq 2e\ln(2k)z_{\ell}&\forall \ell\in B\label{LP4:ell-deg}\\
&\Delta x_j\leq\sum_{\ell\in\Gamma_B(j)}x_j^{\ell}\leq 2\Delta x_j&\forall j\in J \label{LP4:j-deg}\\
&0\leq x_j^{\ell}\leq x_j,z_{\ell}\leq 1&\forall \ell\in B\forall j\in J \label{LP4:j-ell}\\
&\sum_{h\in\Gamma_S(i)}w_h\geq y_i &\forall i\in R\label{LP4:i-cover}\\
%&0\leq w_h\leq y_i\leq 1 &\forall (i,h)\in E(R,S)\\
&0\leq x_j\leq y_i &\forall (j,i)\in E(J,R)\label{LP4:xy}
\end{align}
We further strengthen this LP by partially lifting the above constraints. Specifically, for every $a\in J\cup S$, $j\in J$, $h\in S$, $i\in R$, and $\ell\in B$ we introduce variables $X_h^{(a)}, X_{\ell}^{(a)}, X_{\ell,j}^{(a)}, X_{j}^{(a)}, X_i^{(a)}$, and lift all the above constraints accordingly. For a precise definition, see Appendix~\ref{sec:LP}. For any $j\in J$ such that $x_j>0$ or $h\in S$ such that $w_h>0$, we will denote the ``conditioned" variables by $\hat w^{(j)}_h=X_{h}^{(j)}/x_j$, $\hat z^{(h)}_{\ell}=X_{\ell}^{(h)}/w_h$, etc. %{\color{red} Looks like we actually need to condition on $h\in S$. Need to fix this.}
\begin{remark}
  The above linear program is a relaxation for the partial solution given by Lemma~\ref{lem:mmsa4-partial-solution}. Specifically, given an optimal solution $(J_{\opt},S_{\opt})$, applying the lemma to $J_0=J_{\opt}$, we have the following feasible solution: Set the variables $z_\ell$ and $x_j$ to be indicators for $B'$ and $J'$ as in the lemma, respectively, and the variables $x_{j}^{\ell}$ to be indicators for $J'\times B'$. Set the the variables $w_h$ to be indicators for $S_{\opt}$, and the variables $y_i$ to be indicators for the red neighbors $\Gamma_R(J')$ of $J'$.
\end{remark}

Let us examine some useful properties of this relaxation. First of all, we note that it approximately determines the total LP value of $J$ (since the LP assigns total LP value $\tilde\Theta(|B|)$ to $B$):
\begin{claim}
  Any solution satisfying constraints~\eqref{LP4:ell-deg}-\eqref{LP4:j-ell}, has total LP weight in $J$ bounded by
  $$\frac{1}{2\Delta} \sum_{\ell\in B}z_{\ell}\leq \sum_{j\in J} x_j \leq \frac{2e\ln(2k)}{\Delta}\sum_{\ell\in B}z_{\ell}.$$
  \label{clm:LP4-J-weight}
\end{claim}
\begin{proof}
  \begin{align*}
    \frac{1}{2\Delta} \sum_{\ell\in B}z_{\ell}%&
    \leq\frac{1}{2\Delta}\sum_{\ell\in B}\sum_{j\in \Gamma_{J
    }(\ell)}x_j^{\ell}%\\
    &=\frac{1}{2\Delta}\sum_{j\in J
    }\sum_{\ell\in\Gamma_B(j)}x_j^{\ell}&\text{by~\eqref{LP4:ell-deg}}\\
    &\leq \sum_{j\in J
    } x_j&\text{by~\eqref{LP4:j-deg}}\\
    &\leq \frac{1}{\Delta}\sum_{j\in J}\sum_{\ell\in\Gamma_B(j)}x_j^{\ell}&\text{by~\eqref{LP4:j-deg}}\\
    &%
    = \frac{1}{\Delta}\sum_{\ell\in B}\sum_{j\in \Gamma_{J
    }(\ell)}x_j^{\ell}\leq \frac{2e\ln(2k)}{\Delta}\sum_{\ell\in B}z_{\ell}.&\text{by~\eqref{LP4:ell-deg}}
\end{align*}
\end{proof}

These constraints also determine a useful combinatorial property: in any feasible solution, the number of blue neighbors a subset of $J$ has is (at least) proportional to the LP value of that set.
\begin{claim}
  For any solution satisfying constraints~\eqref{LP4:ell-deg}-\eqref{LP4:j-ell}, and any subset of vertices $\hat J\subseteq J$, the number of blue neighbors of $\hat J$ is bounded from below by:
  $$|\Gamma_B(\hat J)|
  \geq \frac{1
    }{4e\ln(2k)\log k\log m}\cdot \frac{x(\hat J)}{x(J
    )}\cdot |B|.$$
  \label{clm:LP4-blue-neighborhood}
\end{claim}
\begin{proof}
  \begin{align*}
    |\Gamma_B(\hat J)|&\geq\sum_{\ell\in\Gamma_B(\hat J)}z_{\ell}\\
    &\geq \frac{1}{2e\ln(2k)}\sum_{\ell\in\Gamma_B(\hat J)}\sum_{j\in\Gamma%_{J_0}
    (\ell)}x_j^{\ell}&\text{by~\eqref{LP4:ell-deg}}\\
    &\geq \frac{1}{2e\ln(2k)}\sum_{\ell\in\Gamma_B(\hat J)}\sum_{j\in\Gamma_{\hat J}(\ell)}x_j^{\ell}&\text{since }\hat J\subseteq J\\
    &= \frac{1}{2e\ln(2k)}\sum_{j\in\hat J}\sum_{\ell\in\Gamma_B(j)}x_j^{\ell}\\
    &\geq \frac{\Delta}{2e\ln(2k)}\sum_{j\in\hat J
    }x_j&\text{by~\eqref{LP4:j-deg}}\\
    &\geq\frac{1}{2e\ln(2k)}\cdot \frac{z(B)}{2x(J%_0
    )}\cdot x(\hat J_0)&\Delta\geq \frac{z(B)}{2x(J)}\text{ by Claim~\ref{clm:LP4-J-weight}}\\
    &=\frac{1}{4e\ln(2k)}\cdot \frac{x(\hat J_0)}{x(J%_0
    )}\cdot z(B)\\
    &\geq \frac{1%\rho
    }{4e\ln(2k)\log k\log m}\cdot \frac{x(\hat J)}{x(J%_0
    )}\cdot |B|&\text{by~\eqref{LP4:B-weight}}
\end{align*}
\end{proof}

A fractional variant of the above covering property for the blue vertices is the following:
\begin{claim}
  For any solution satisfying constraints~\eqref{LP4:ell-deg}-\eqref{LP4:j-ell}, and any subset of vertices $\hat J\subseteq J$, at least $\eps|B|$ vertices $\ell\in B$ satisfy
  $$\sum_{j\in\Gamma_{\hat J}(\ell)}x_j^{\ell}\geq \frac{1}{4\log k\log m} \cdot\frac{x(\hat J)}{x(J)},$$
    where
    $$\eps=\frac{1}{8e\ln(2k)\log k\log m}\cdot \frac{x(\hat J)}{x(J)}.$$
  \label{clm:LP4-blue-neighborhood-frac}
\end{claim}
\begin{proof}
  Let $$B_0=\left\{\ell\in B\;\left| \sum_{j\in\Gamma_{\hat J}(\ell)}x_j^{\ell}\geq \frac{1}{4\log k\log m} \cdot\frac{x(\hat J)}{x(J)}\right.\right\}.$$ We need to show that $|B_0|\geq\eps|B|$. Assume for the sake of contradiction that $|B_0|<\eps|B|$. On the one hand, we know that
  \begin{align*}
      \sum_{\ell\in B}\sum_{j\in\Gamma_{\hat J}(\ell)}x^j_{\ell}&=\sum_{j\in \hat K}\sum_{\ell\in\Gamma_B(\ell)}x^j_{\ell}\\
      &\geq\Delta\cdot\sum_{j\in\hat J}x_j&\text{by~\eqref{LP4:j-deg}}\\
      &=\Delta x(J)\cdot\frac{x(\hat J)}{x(J)}\\
      &\geq \frac{1}{2}\sum_{j\in \hat J}\sum_{\ell\in\Gamma_B(\ell)}x^j_{\ell}\cdot \frac{x(\hat J)}{x(J)}&\text{by~\eqref{LP4:j-deg}}\\
      &=\frac{1}{2}\sum_{\ell\in B}\sum_{j\in\Gamma_J(\ell)}x^j_{\ell}\cdot \frac{x(\hat J)}{x(J)}\\
      &\geq \frac12\sum_{\ell\in B}z_\ell\cdot \frac{x(\hat J)}{x(J)}&\text{by~\eqref{LP4:ell-deg}}\\
      &\geq \frac{|B|}{2\log k\log m}\cdot \frac{x(\hat J)}{x(J)}.&\text{by~\eqref{LP4:B-weight}}
  \end{align*}
  On the other hand, by our assumption, we have
  \begin{align*}
      \sum_{\ell\in B}\sum_{j\in\Gamma_{\hat J}(\ell)}x^j_{\ell}&=\sum_{\ell\in B_0}\sum_{j\in\Gamma_{\hat J}(\ell)}x^j_{\ell} +\sum_{\ell\in B\setminus B_0}\sum_{j\in\Gamma_{\hat J}(\ell)}x^j_{\ell}\\
      &<\sum_{\ell\in B_0}\sum_{j\in\Gamma_{\hat J}(\ell)}x^j_{\ell}+\frac{|B|}{4\log k\log m} \cdot\frac{x(\hat J)}{x(J)}&\text{by def.\ of }B_0\\
      &\leq 2e\ln(2k)\sum_{\ell\in B_0}z_\ell+\frac{|B|}{4\log k\log m} \cdot\frac{x(\hat J)}{x(J)}&\text{by~\eqref{LP4:ell-deg}}\\
      &\leq |B_0|\cdot 2e\ln(2k)+\frac{|B|}{4\log k\log m} \cdot\frac{x(\hat J)}{x(J)}\\
      &<\frac{|B|}{2\log k\log m} \cdot\frac{x(\hat J)}{x(J)},&\text{assuming }|B_0|< \eps|B|
  \end{align*}  
  which contradicts the previous bound. Thus we must have $|B_0|\geq\eps|B|$ as required.
\end{proof}

\begin{algorithm}[p]
    \caption{Approximation Algorithm for MMSA$_4$}
    \label{alg-mmsa4}
\begin{algorithmic}
\State \textbf{Input:} $B, J, R, S, E$
    \State{\textbf{Guess} $\opt,\Delta$ and solve the LP} \Comment{e.g. using binary search}
    %\State{$J_\alg = \emptyset$}\Comment{$J_{\alg}$ stores the current solution}
    \State{\textbf{Choose} parameter $s$ such that the LP weight of the bucket $J_s=\{j\in J\mid 2^{-s}\leq x_j\leq 2^{-(s-1)}\}$, that is, $\sum_{j\in J_s}x_j$ is maximized, and let $x_0=2^{-s}$ and $J_0=J_s$.}
    \For{every $j\in J_0$ and $i\in \Gamma_R(j)$}
       \State{\textbf{Choose} a new parameter $s$ such that the conditioned LP weight of the bucket\\ $S^{ji}_s=\{h\in S\mid 2^{-s}\leq \hat w_h^{(j)}\leq 2^{-(s-1)}\}$, that is, $\sum_{h\in S^{ji}_s}\hat w_h^{(j)}$, is maximized, and let $\beta_{ji}=2^{-s}$.}
       \State{\textbf{Choose} parameter $t$ such that the sub-bucket $\hat S^{ji}_t=\{h\in S^{ji}_s\mid 2^{-t}\leq w_h\leq 2^{-(t-1)}\}$ has maximum cardinality $|\hat S^{ji}_t|$, and let $\gamma_{ji}=2^{-t}$ and $\hat\Gamma_j(i)=\hat S^{ji}_t$.}
    \EndFor
    \For{every $j\in J_0$ and $\beta,\gamma$}
       \State{\textbf{Let} $\Gamma^R_{\beta,\gamma}(j)=\{i\in\Gamma_R(j)\mid \beta_{ji}=\beta,\gamma_{ji}=\gamma\}$.}
       \State{\textbf{Let} $\Gamma^S_{\beta,\gamma}(j)=\bigcup_{i\in\Gamma^R_{\beta,\gamma}(j)}\hat\Gamma_j(i)$.}
    \EndFor
    \For{every $\beta,\gamma$ and $D\in\{2^{s-1}\mid s\in\lceil\log|S|\rceil\}$}
       \State{\textbf{Let} $J^D_{\beta,\gamma}=\{j\in J_0\mid \Gamma^R_{\beta,\gamma}(j)\neq\emptyset,|\Gamma^S_{\beta,\gamma}(j)|\in[D,2D]\}$.}
        \State{\textbf{Let} $T^D_{\beta,\gamma}=\{\langle j, \Gamma^R_{\beta,\gamma}(j),\Gamma^S_{\beta,\gamma}(j)\rangle\mid j\in J^D_{\beta,\gamma}(j)\}$.}
    \EndFor
    \State{\textbf{Let} $P_1=\{\langle\beta,\gamma,D\mid \beta/\gamma>A, \beta D> A\cdot\opt/x(J_0)\rangle\}$, and
    $$J_1=\bigcup_{\langle\beta,\gamma,D\rangle\in P_1}%\bigcup_{\beta,\gamma: \beta/\gamma >A}\;\bigcup_{D>A\cdot\opt/(\beta\cdot x(J_0))}%{\substack{\beta,\gamma,D:%\\
    %\beta/\gamma>A,\\ \beta D > %A\cdot\opt/x(J_0)}}
    J^D_{\beta,\gamma}.$$}
    \If{$|J_1|<|J_0|/2$}
      \State{\textbf{Let} $J_{\alg}=\emptyset$.}
      \For{all $j\in J_0\setminus J_1$}
        \State{Independently add $j$ to $J_{\alg}$ with probability $x_0$.}
      \EndFor
      \State{\textbf{Let} $S_{\alg}=\emptyset$.}
      \For{every $\beta$}
        \State{\textbf{Let} $S_{\beta}=\bigcup_{j\in J_{\alg}}\bigcup_{\gamma}\Gamma^S_{\beta,\gamma}(j)$.}
          \For{all $h\in S_{\beta}$}
            \State{Independently add $h$ to $S_{\alg}$ with probability $\min\{1,\beta\cdot 12\log|S|\log(|S|^2m)\ln n\}$.}%\Comment{{\color{red}Need to be more specific.}}
        \EndFor
      \EndFor
    \ElsIf{$|J_1|\geq|J_0|/2$}
      \State{\textbf{Choose }$\langle\beta,\gamma,D\rangle\in P_1$ that maximize the cardinality $|J_{\beta,\gamma}^D|$, and let $J_2=J_{\beta,\gamma}^D$.}
      \State{\textbf{Let } $S_{\tilde D}=\{h\in S\mid\{j'\in J_2\mid h\in \Gamma_{\beta,\gamma}^S(j')\}|\in[\tilde D,2\tilde D]\}$ for every $\tilde D\in \{2^{s-1}\mid s\in\lceil\log|J_2|\rceil\}$.}
      \State{\textbf{Choose }$\tilde D$ that maximizes the cardinality $|\{\langle j,h\rangle\in J_2\times S_{\tilde D}\mid h\in \Gamma_{\beta,\gamma}^S(j)\}|$, and let $\tilde S=S_{\tilde D}$.}
      \State{\textbf{Choose }$h_0\in \tilde S$ that maximizes the total LP value $\sum_{j\in J_2:\Gamma_{\beta,\gamma}^S(j)\ni h_0}\hat x_j^{(h_0)}$.} %, and let $\tilde J=\{j\in J_2\mid h_0\in\Gamma_{\beta,\gamma}^D(j)\}$.}
      %\For{all $j\in \tilde J$}
       % \State{Independently add $j$ to $J_{\alg}$ with probability $\beta x_0/w_{h_0}$.}
      %\EndFor
      \State{\textbf{Let} $J_{\alg}=\{j\in J_2\mid h_0\in\Gamma_{\beta,\gamma}^S(j)\}$.}
      \State{\textbf{Let} $S_{\alg}=\emptyset$.}
      \For{every $h\in \bigcup_{j\in J_{\alg}}\Gamma_S(\Gamma_R(j))$}
        \State{Independently add $h$ to $S_{\alg}$ with probability $\min\{1,\hat w^{h_0}_h\cdot 4\gamma\ln n/(x_0\beta)\}$.}%\Comment{{\color{red}Need to be more specific.}}
      \EndFor
    \EndIf
\end{algorithmic}
\end{algorithm}

\iffalse
The algorithm involves a number of bucketing steps. The following property will be useful in analyzing these steps:

\begin{claim} Let $\{\alpha_c\mid c\in C\}$ be a finite sequence of values in the range $[0,1]$, where the cardinality of the sequence is bounded by $|C|\leq K$ for some $K\geq 4$. Then %, letting $\eps=\min\{\sum_{c\in C}\alpha_c, 1\}$ then
$$\max_{s\in\nats}\sum_{c\in C:\alpha_c\in[2^{-c},2^{-(c-1)}]}\alpha_c\geq\frac{\sum_{c\in C}\alpha_c}{3\log K},$$
and moreover any $s$ attaining this maximum must satisfy $s\leq \log(2K\log K)\leq 2\log K$.
  \label{clm:bucketing}
\end{claim}
\fi

Let us now analyze the approximation guarantee of Algorithm~\ref{alg-mmsa4}. %{\color{red} Need to mention somewhere that this is an algorithm for ``making progress'' towards a good approximation.}
We begin by stating simple lower bounds on the total LP value of the set $J_0$ as well as the vertices in the set.

\begin{lemma}
  The set $J_0$ defined in Algorithm~\ref{alg-mmsa4} has LP value at least $x(J)/(2\log m)$ and the lower bound on the individual LP values in the set is bounded by $x_0\geq 1/m$.
  \label{lem:J0-x0-bounds}
\end{lemma}
\begin{proof}
  By our simplifying assumptions, we have $\Delta\leq k/m^{\eps}$, and therefore, by Claim~\ref{clm:LP4-J-weight}, the set $J$ has total LP weight at least
  \begin{align*}
      \sum_{j\in J}x_j&\geq \frac{1}{2\Delta}\sum_{\ell\in B}z_\ell\\
      &\geq \frac{k}{\Delta\log k\log m}&\text{by~\eqref{LP4:B-weight}}\\
      &\geq \frac{m^{\eps}}{\log k\log m},
  \end{align*}
  and so in particular, $x(J)=\omega(1)$ and the total LP weight of vertices in $J$ with LP value at most $1/m$ is bounded by $\sum_{j\in J:x_j\leq 1/m}x_j\leq 1\leq x(J)/2$, and there is some $s\in[\log m]$ such that $J_s$ has LP weight at least $x(J)/(2\log m)$, which gives our lower bound on $J_0$ (the heaviest bucket). Moreover, we know that the heaviest bucket can't be $J_s$ for $2^{-s}<1/m$, since the total LP weight of all vertices with at most this LP value is at most $\tilde O(x(J)/m^{\eps})$. Thus, $x_0\geq 1/m$.
\end{proof}

Next, we examine the bucketing of neighbors in $S$, and give a lower bound on the number of vertices in these bucketed sets.

\begin{lemma}
  In Algorithm~\ref{alg-mmsa4}, for every vertex $j\in J_0$, and every red neighbor $i\in\Gamma_R(j)$, the bucketed set of neighbors $\hat\Gamma_j(i)$ of $i$ has cardinality bounded from below by $|\hat\Gamma_j(i)|\geq1/(6\beta_{ji}\log|S|\log(|S|^2m))$.
  \label{lem:beta-ji-size}
\end{lemma}
\begin{proof}
  Fix vertices $j\in J_0$ and $i\in\Gamma_R(j)$. Let us begin by examining our choice of $\beta_{ji}$. Note that lifting Constraint~\ref{LP4:xy}, we get $x_j=X^{(j)}_j\leq X^{(j)}_i(\leq x_j)$, and so $\hat y^{(j)}=X^{(j)}_i/x_j=1$. Lifting Constraint~\eqref{LP4:i-cover}, we thus get
  $$\sum_{h\in\Gamma_S(i)}\hat w^{(j)}_h\geq 1.$$
  Note that the total LP weight of the set $S'_{ji}=\{h\in\Gamma_S(i)\mid \hat w^{(j)}_h\leq 1/|S|^2\}$ is at most $1/|S|\leq \hat w^{(j)}(\Gamma_S(i))/3$. Therefore, the total $\hat w^{(j)}$ LP weight of the bucketed sets $S_s^{ji}$ for $s$ such that $2^{-s}\geq 1/|S|^2$ is at least $\frac23\hat w^{(j)}(\Gamma_S(i))$, and at least one of these bucketed sets has LP weight at least a $1/(2\log|S|)$-fraction of this, or at least $\hat w^{(j)}(\Gamma_S(i))/(3\log|S|)\geq 1/(3\log|S|)$. This gives a lower bound on the LP weight of the bucket which defines $\beta_{ji}$. Also, the heaviest bucket cannot be $S_s^{ji}$ for $s$ such that $2^{-s}\leq 1/|S|^2$, since even the total weight of these buckets is at most $1/|S|=o(1/(3\log|S|))$. In particular, this means that $\beta_{ji}\geq 1/|S|^2$. Moreover, for $s$ such that $2^{-s}=\beta_{ji}$, since the total conditional LP weight of $S^{ji}_s$ is at least $1/(3\log|S|)$, and every vertex in the set has conditional LP value at most $2\beta_{ji}$, the cardinality of the set must be at least $1/(6\beta_{ji}\log|S|)$.
  
  Now let us examine the second stage of bucketing. Note that for every $h\in\Gamma_S(i)$, we have
  $$w_h\geq X^{(j)}_h=\hat w^{(j)}_h\cdot x_j\geq \beta_{ji}x_0\geq 1/(|S|^2m)$$ (and, of course, $w_h\leq 1$). Therefore, the number of non-empty buckets $\hat S^{ji}_t$ is at most $\log(|S|^2m)$, and at least one of them must have cardinality at least $|S^{ji}_s|/\log(|S|^2m)$, which, along with our lower bound on $|S^{ji}_s|$ above, gives us the required lower bound on $|\hat\Gamma_j(i)|$.
\end{proof}

Note that from the above proof, we also get upper-bounds on the number of values of $\beta_{ji}$ and $\gamma_{ji}$ that can produce non-empty buckets. In particular, we get the following bound:
\begin{observation}
  The total number of possible values for $\beta_{ji}$ is at most $2\log|S|$, and the total number of possible values for $\gamma_{ji}$ is at most $\log(|S|^2m)$. Along with the range of values for $D$, the total number of triples $\langle \beta,\gamma,D\rangle$ for which $J^{D}_{\beta,\gamma}$ is non-empty is at most $2\log^2|S|\log(|S|^2m)$.\label{obs:triple-bound}
\end{observation}

The algorithm proceeds by separating the buckets corresponding to parameters for which the simple rounding (which samples a random subset of $J$ of size $\tilde\Omega(J)$) makes progress towards an approximation guarantee of $\tilde O(A)$. If a large fraction of vertices in $J_0$ participate exclusively in such buckets, then the algorithm applies this rounding. The following lemma gives the analysis of the algorithm in this case.

\begin{lemma}
  In Algorithm~\ref{alg-mmsa4}, if $|J_1|<|J_0|/2$, then with high probability the algorithm samples a subset $J_{\alg}\subseteq J$ which covers an $\tilde\Omega(1)$-fraction of blue vertices, and a subset of $S$ of size $\tilde O(A\cdot\opt)$ which covers all the red neighbors of $J_{\alg}$. \label{lem:mmsa4-case-1}
\end{lemma}
\begin{proof}
  Let us begin by analyzing the number of blue vertices covered. First, we can bound the LP value of the set $J_0\setminus J_1$ by
  \begin{align*}
      x(J_0\setminus J_1&)=x(J)-x(J_1)&\text{since }J_1\subseteq J_0\\
      &=x(J)-\frac{x(J_1)}{|J_1|}\cdot |J_1|\\
      &\geq x(J)-2\cdot\frac{x(J_0\setminus J_1)}{|J_0\setminus J_1|}\cdot |J_1|&\text{since }\forall j\in J:x_j\in[x_0,2x_0]\\
      &>x(J)-2\cdot\frac{x(J_0\setminus J_1)}{|J_0\setminus J_1|}\cdot |J_0\setminus J_1|&\text{since }|J_1|<|J_0|/2\\
      &= x(J)-2x(J_0\setminus J_1)\qquad \Longrightarrow\qquad x(J_0\setminus J_1)> x(J_0)/3.
  \end{align*}
  In particular, by Lemma~\ref{lem:J0-x0-bounds}, we get that $x(J_0\setminus J_1)>x(J)/(6\log m)$. Thus, by Claim~\ref{clm:LP4-blue-neighborhood-frac}, there is a blue subset of $B_0\subseteq B$ of cardinality
  $$|B_0|\geq \frac{|B|}{48e\ln(2k)\log k\log^2 m},$$
  such that for every $\ell\in B_0$, we have
  $$\sum_{j\in\Gamma_{J\setminus J_0}(\ell)}x_j^{\ell}\geq\frac{1}{24\log k\log^2m}.$$
  Since for every neighbor $j$ above we have $x_j^{\ell}\leq x_j\leq 2x_0$, we get that every $\ell\in B_0$ has at least
  $$|\Gamma_{J_0\setminus J_1}(\ell)|\geq \frac{1}{48\log k\log^2m}\cdot \frac{1}{x_0}$$
  neighbors in $J_0\setminus J_1$. Thus, when the algorithm samples every vertex $j\in J_0\setminus J_1$ independently with probability $x_0$, with high probability this covers at least
  $$\Omega\left(\frac{1}{\log k\log^2m}\right)\cdot|B_0|=\Omega\left(\frac{|B|}{\log^3k\log^4m}\right)$$ blue vertices.
  
  Before analyzing the cost incurred by the algorithm, let us note that it outputs a feasible solution with high probability. Indeed, by Lemma~\ref{lem:beta-ji-size}, for every $j\in J_{\alg}$ and every red neighbor $i\in\Gamma_(j)$, the red vertex $i$ has at least $1/(6\beta_{ji}\log|S|\log(|S|^2m))$ neighbors in $S$, and by definition all of them belong to $S_\beta$ as defined at this point in the algorithm. Therefore, every such neighbor $i$ is covered with probability at least $1-1/n^2$, and with high probability, all red neighbors of $J_{\alg}$ are covered.
  
  We now turn to analyzing the cost $|S_{\alg}|$ of this solution. By definition of $J_1$, for any set $J_{\beta,\gamma}^D$ that intersects $J_0\setminus J_1$, we must have $\beta/\gamma\leq A$ or $\beta D\leq A\cdot\opt/(x(J_0))$. Equivalently, for any $j\in J_0\setminus J_1$, if $\Gamma^S_{\beta,\gamma}(j)\neq\emptyset$, then either
  \begin{align}
      &\frac{\beta}{\gamma}\leq A,\qquad\qquad\text{or}\label{J01-cond-1}\\
      &\beta |\Gamma^S_{\beta,\gamma}(j)|\leq\frac{2A\cdot\opt}{x(J_0)}.\label{J01-cond-2} 
  \end{align}
  And so for every $\beta$, we can divide the set $S_\beta$ into two (potentially overlapping) parts $S_\beta=S_\beta^1\cup S_\beta^2$, where
  $$S_{\beta}^1=\bigcup_{j\in J_{\alg}}\bigcup_{\gamma:\beta/\gamma\leq A}\Gamma^S_{\beta,\gamma}(j)\qquad\text{and}\qquad S_{\beta}^2=\bigcup_{j\in J_{\alg}}\bigcup_{\gamma:\beta |\Gamma^S_{\beta,\gamma}(j)|\leq 2A\cdot\opt/x(J_0)}\Gamma^S_{\beta,\gamma}(j).$$
  We will analyze the number of vertices sampled from these two parts separately.
  
  Beginning with $S_\beta^1$, for every $\gamma$ which satisfies~\eqref{J01-cond-1}, consider the subset
  $$S_{\beta,\gamma}^1=\bigcup_{j\in J_{\alg}}\Gamma^S_{\beta,\gamma}(j).$$ Note that the union of these subsets is exactly $S_\beta^1$. Note that for every vertex $h\in S_{\beta,\gamma}^1$, we have $w_h\geq \gamma$, and so since the total LP value is at most $\opt$, we have that $|S_{\beta,\gamma}^1|\leq \opt/\gamma$. With high probability, the number of vertices sampled from this set (in the iteration corresponding to $\beta$) is at most $O(\beta\cdot\log(S)\log(|S|m)\log n\cdot \opt/\gamma)$, and since by Observation~\ref{obs:triple-bound}, the number of $\langle \beta,\gamma\rangle$ pairs that can contribute to $S_{\alg}$ is at most $O(\log|S|\log(|S|m)$, we get a total upper bound on this contribution to $S_\alg$ of 
  $$O(\beta\cdot\log^2|S|\log^2(|S|m)\log n\cdot \opt/\gamma)\leq O(A\cdot\opt\cdot \log^2|S|\log^2(|S|m)\log n),$$
  where the inequality follows from~\eqref{J01-cond-1}.
  
  Now let us examine the contribution from the various $S_\beta^2$. Note that for any $j\in J_0\setminus J_1$ and $\beta,\gamma$ that satisfy~\eqref{J01-cond-2}, the expected number of vertices that $\Gamma^S_{\beta,\gamma}$ will contribute to $S_{\alg}$ is at most
  $$O\left(|\Gamma^S_{\beta,\gamma}(j)\cdot \beta\cdot \log|S|\log(|S|m)\log n|\right)\leq O\left(\frac{A\cdot \opt}{x(J_0)}\cdot \log|S|\log(|S|^2m)\log n\right),$$
  and since with high probability we have $|J_{\alg}|=O(x(J_0))$, combined with the above bound on possible contributing pairs $\langle \beta,\gamma\rangle$, we get a total bound of
  $$O(A\cdot\opt\cdot \log^2|S|\log^2(|S|m)\log n)$$ as in the previous contribution to $S_{\alg}$, which concludes the proof of the upper bound on $|S_{\alg}|$. %Together with our lower bound on the number of blue vertices covered, we get that in this case, the algorithm makes progress towards an approximation ratio of
  %$$O\left(A\cdot \log^2|S|\log^2(|S|m)\log n\cdot \log^3k\log^4 m\cdot \right).$$
\end{proof}

Finally, we turn to the remaining case in Algorithm~\ref{alg-mmsa4}, when $|J_1|\leq |J_0|/2$. The analysis of this case rests on a back-degree argument similar (though significantly more involved) to the argument in Lemma~\ref{lem:back-deg} for Red-Blue Set Cover. Indeed, we show the following:

\begin{lemma}
  If $|J_1|\geq |J_0|/2$, then for $\beta,\gamma,D,h_0$ and the set $J_{\alg}$ as defined by the algorithm in this case, we have
  $$\sum_{j\in J_{\alg}}\hat x_j^{(h_0)}\geq\frac{|J_0|Dx_0\beta}{\opt}\cdot\frac{1}{4\log^2|S|\log(|S|^2m)\log m}.$$
  Furthermore, for every vertex $j\in \tilde J$ (as defined by the algorithm), we have $\hat x_j^{h_0}\in[x_0\beta/(2\gamma),4x_0\beta/\gamma]$.
  \label{lem:mmsa4-back-deg}
\end{lemma}
\begin{proof}
  We begin with the second claim. By definition of $S_{\hat D}$ and $\tilde J$, for every $j\in \tilde J$  we have $w_{h_0}\in[\gamma,2\gamma]$, $\hat w^{(j)}_h\in[\beta,2\beta]$ and of course $x_j\in[x_0,2x_0]$. Since by definition we have $\hat w^{(j)}_{h_0}=X_{h_0}^{(j)}/x_j$ and $\hat x^{(h_0)}_j/=X^{(h_0)}_j/w_{h_0}=X^{(j)}_{h_0}/w_{h_0}$, we have that $\hat x_j^{(h)}=x_j\hat w_h^{(j)}/w_h$, and then the bounds on $\hat x_j^{(h)}$ follow immediately from the respective bounds on $x_j$, $w_h$, and $\hat w_h^{(j)}$.
  
  Now let us show the lower bound on the conditional LP weight of $J_{\alg}$. We begin by noting that by our choice of $\langle\beta,\gamma,D\rangle$, Observation~\ref{obs:triple-bound} and the current case in the algorithm, the cardinality of the set $J_2$ can be bounded by
  $$|J_2|\geq\frac{|J_1|}{|P_1|}\geq\frac{|J_0|}{2|P_1|}\geq\frac{|J_0|}{4\log^2|S|\log(|S|^2\log m)}.$$ Now consider the set of pairs $E_1=\{\langle j,h\rangle\in J_2\times S\mid h\in\Gamma^S_{\beta,\gamma}(j)\}$. Since $J_2\subseteq J^{D}_{\beta,\gamma}$, we can bound the number of pairs in this set by $|E_1|\geq|J_2|D$. Thus, by our choice of $\tilde D$, we have
  $$|\{\langle j,h\rangle\in J_2\times \tilde S\mid h\in\Gamma^S_{\beta,\gamma}(j)\}|\geq \frac{|J_2|D}{\log m}.$$
  We can now use an averaging argument to bound the conditional LP value of $S_{\alg}$. First note by the bounds on $x_0$ and $\hat x_j^h$ observed above that we have
  \begin{align*}
      \sum_{h\in \tilde S}w_h\sum_{j\in J_2:\Gamma^S_{\beta,\gamma}(j)\ni h}\hat x_j^{(h)}&\geq |\{\langle j,h\rangle\in J_2\times \tilde S\mid h\in\Gamma^S_{\beta,\gamma}(j)\}|x_0\beta\\
      &\geq \frac{|J_2|Dx_0\beta}{\log m}.
  \end{align*}
  On the other hand, we know that $\sum_{h\in\tilde S}w_h\leq\opt$, and so there must be some $h\in\tilde S$ for which
  \begin{align*}
      \sum_{j\in J_2:\Gamma^S_{\beta,\gamma}(j)\ni h}\hat x_j^{(h)} &\geq \frac{|J_2|Dx_0\beta}{\opt}\cdot\frac{1}{\log m}\\
      &\geq \frac{|J_0|Dx_0\beta}{\opt}\cdot\frac{1}{4\log^2|S|\log(|S|^2\log^2m)}.&\text{by our bound on }J_2
  \end{align*}
  In particular, our choice of $h_0$ must satisfy this property.
\end{proof}

We can now show that in this case, the algorithm makes progress towards an $\tilde O(m/A^2)$-approximation. Trading this off with the progress towards an $\tilde O(A)$-approximation as guaranteed by Lemma~\ref{lem:mmsa4-case-1}, we get an $\tilde O(m^{1/3})$-approximation by setting $A=m^{1/3}$.

\begin{lemma}
  In Algorithm~\ref{alg-mmsa4}, if $|J_1|\geq |J_0|/2$, then with high probability, the algorithm makes progress towards an approximation guarantee of $\tilde O(m/A^2)$.
  \label{lem:mmsa4-case-2}
\end{lemma}
\begin{proof}
  Let us first bound the number of blue vertices covered by $J_\alg$. By Lemma~\ref{lem:mmsa4-back-deg}, we have
  \begin{align*}
      \hat x^{h_0}(J_\alg)&\geq \frac{|J_0|x_0D\beta}{\opt}\cdot\frac{1}{4\log^2|S|\log(|S|^2m)\log m}\\
      &\geq \frac{x(J_0)D\beta}{\opt}\cdot\frac{1}{4\log^2|S|\log(|S|^2m)\log m}&\text{since }\forall j\in J_0:2x_0\geq x_j\\
      &\geq \frac{x(J)D\beta}{\opt}\cdot\frac{1}{4\log^2|S|\log(|S|^2m)\log m}&\text{by Lemma~\ref{lem:J0-x0-bounds}}\\
      &>x(J)\cdot\frac{A}{x(J_0)}\cdot \frac{1}{4\log^2|S|\log(|S|^2m)\log m}.&\text{since }\forall\langle\beta,\gamma,D\rangle\in P_1:\frac{\beta D}{\opt}>\frac{A}{x(J_0)}
  \end{align*}
  Thus, since the conditioned LP solution satisfies the basic LP, we can apply Claim~\ref{clm:LP4-blue-neighborhood} to this solution and get that the size of the blue neighborhood of $J_{\alg}$ can be bounded by
  \begin{align*}
      |\Gamma_B(J_{\alg})|&\geq \frac{1
    }{4e\ln(2k)\log k\log m}\cdot \frac{A}{x(J_0)}\cdot\frac{1}{4\log^2|S|\log(|S|^2m)\log m}\cdot |B|\\
    &=\frac{1}{16\ln(2k)\log k\log^2|S|\log(|S|^2m)\log^2m}\cdot\frac{A}{x(J_0)}\cdot|B|.
  \end{align*}
  
  Note that by the LP constraints and Lemma~\ref{lem:mmsa4-back-deg}, for every red neighbor $i\in\Gamma_R(J_{\alg})$, we have
  $$\hat y^{h_0}_i\geq x^{h_0}_i\geq x_0\beta/(2\gamma),$$ and so by~\eqref{LP4:i-cover}, the rescaled solution $(\hat w^{h_0}_h\cdot 2\gamma/(x_0\beta))_{h\in S}$ is a fractional set cover for $\Gamma_R(J_{\alg})$. Thus, sampling every $h\in S$ with probability $\min\{1,2\ln n\cdot \hat w^{h_0}_h\cdot 2\gamma/(x_0\beta)\}$ produces a valid set cover with high probability. It remains to analyze the size of this set cover. Indeed, since $\hat w^{h_0}(S)\leq \opt$, our sampling procedure produces a set of expected size
  \begin{align*}
      \expec[|S_{\alg}|]&\leq \frac{4\gamma\ln n}{x_0\beta}\cdot\opt\\
      &\leq 8\ln n\cdot\frac{\gamma}{\beta}\cdot\frac{|J_0|}{x(J_0)}\cdot\opt&\text{since }\forall j\in J_0:x_j\leq 2x_0\\
      &\leq 8\ln n\cdot\frac{\gamma}{\beta}\cdot\frac{m}{x(J_0)}\cdot\opt\\
      &< 8\ln n\cdot\frac{1}{A}\cdot\frac{m}{x(J_0)}\cdot\opt,&\text{since }\forall\langle\beta,\gamma,D\rangle\in P_1:\frac{\beta}{\gamma}>A
  \end{align*}
  and so with high probability we have $|S_{\alg}|=O(\opt\cdot m\log n/(A\cdot x(J_0)))$.
 
  Putting our two bounds together, we get that in this case, the algorithm makes progress towards an approximation guarantee of
  $$\frac{|B|}{|\Gamma_B(J_{\alg})|}\cdot\frac{|S_{\alg}|}{\opt}=\tilde O(1)\cdot \frac{x(J_0)}{A}\cdot \frac{m}{A\cdot x(J_0)}=\tilde O(1)\cdot \frac{m}{A^2}.$$
\end{proof}

\section{Sketch of Integrality Gap for MMSA\texorpdfstring{$_{t}$}{t}}\label{sec:lb-mmsa}
In this section we provide a proof sketch of integrality gap for MMSA$_t$ by gluing together hard instances of Densest $k$-Subgraph as the layers in MMSA$_t$.
Our lower bound is based on the current best known bounds for Densest $k$-Subgraph in different parameter regimes of $k$. Consider the Erd\"os--R\'enyi graph $G(n,p)$ where $p=n^{\beta-1}$, and let $k=n^{\alpha}$ for some $\alpha,\beta\in(0,1)$. As we know, log-density-based algorithms (and Sherali-Adams) can detect the presence of a planted $k$-subgraph of average degree $n^{\gamma}$ for any $\gamma>\alpha\beta$, while on the other hand, for $\alpha\leq 1/2$, even Sum-of-Squares cannot certify the non-existence of such a subgraph for $\gamma<\alpha\beta$~\citep{chlamtavc2018sherali}. For $\alpha>1/2$, however, the situation is different. While the limit of Sherali-Adams and log-density techniques is still $\gamma>\alpha\beta$, a simple SDP relaxation (as well as other simple techniques) can certify the non-existence of dense subgraph for $\gamma>\beta/2$.

Thus, we can summarize the limits of current techniques (as well as the integrality gap of Sherali-Adams on top of a simple SDP) as: In $G(n,p)$ for $p=n^{\beta-1}$, for $k=n^{\alpha}$ we can certify the non-existence of a planted $k$-subgraph of average degree $\omega(n^{\gamma})$ iff $\gamma\geq\min\{\alpha,1/2\}\cdot\beta$.

\paragraph{Construction of the integrality gap instance.} Let $t$ be a positive odd integer, and let $\eps>0$. We construct a layered graph (representing an alternating circuit of depth $t$) as follows: Let $G_1,\ldots, G_{(t-1)/2}$ be independent instances of $G(n,p)$ with $p=n^{-(2-2\eps)/(2-\eps)}=n^{\beta-1}$ for $\beta=\eps/(2-\eps)$. For every $i\in[(t-1)/2]$, let $U_{2i}$ be a vertex set corresponding to the edges of $G_i$, and let $U_{2i+1}$ be a vertex set corresponding to the vertices of $G_i$. For every edge $e=(u,v)$ in $G_i$, add edges between the vertex representing the edge $e$ in $U_{2i}$, and the vertices representing $u$ and $v$ in $U_{2i+1}$. Let $U_1$ be a vertex set of size $n^{(2-2\eps)/(2-\eps)}$. Finally, for every $i\in[(t-1)/2]$, partition every vertex set $U_{2i}$ into $|U_{2i-1}|$ sets arbitrarily (thus, into $n^{(2-2\eps)/(2-\eps)}$ sets for $i=1$, and into $n$ sets for $i>1$), and connect the $j$th vertex in $U_{2i-1}$ to all vertices into the $j$th set in the partition of $U_{2i}$. See Figure~\ref{fig:lb_instance} for an example construction of this lower bound instance.
\begin{figure}[!h]
    \centering
    \includegraphics[width=\textwidth]{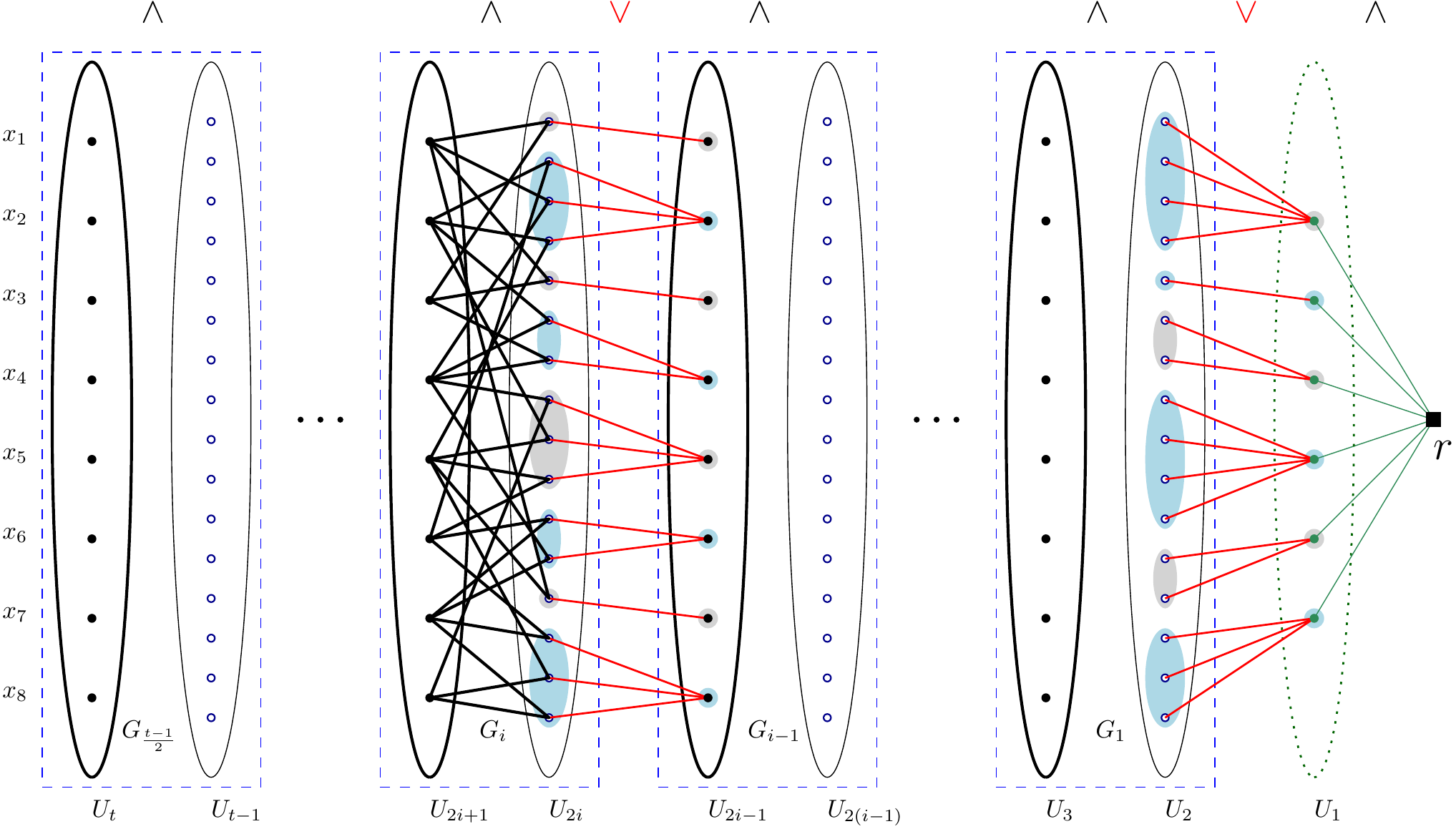}
    \caption{This figure provides a high-level construction of the hardness instance. The bipartite representation of $(t-1)/2$ graphs $G_1, \cdots, G_{\frac{t-1}{2}}$ along with a target set $U_1$ form the layers of our instance of MMSA$_{t}$. Each $G_i$ is represented as a bipartite graph with two set of nodes $U_{2i} = V(G_i)$ and $U_{2i+1} = E(G_i)$ where there is an edge between $e \in U_{2i}$ and $v\in U_{2i+1}$ iff $v$ is one of the endpoints of $e$ in $G_i$ (these are the bold black edges and correspond to AND gates). Moreover, layers of different $G_i$ are glued together via a mapping of edge-nodes of $G_i$ (i.e, the vertices in $U_{2i}$) to vertex-nodes of $G_{i-1}$ (i.e., the vertices in $U_{2i-1}$) and $U_1$ is added to handle the connection for $G_1$. These intermediate edges are the red edges in this figure and correspond to OR gates.}
    \label{fig:lb_instance}
\end{figure}

Note that the total number of vertices is very strongly concentrated around $n^{2/(2-\eps)}\cdot(t-1)/2$. Let us denote $m=n^{2/(2-\eps)}\cdot(t-1)/2$. Since we think of $t$ as a constant, we will analyze the lower bound as a function of $m$ (the total number of vertices up to a constant factor). To that end, let us note that in all the graphs $G_i$, we have $n\approx m^{1-\eps/2}$, and $p=n^{-(2-2\eps)/(2-\eps)}\approx m^{-(1-\eps)}$.

\paragraph{Proof sketch of integrality gap.} Let us begin by analyzing the optimum. We need to cover all $|U_1|\approx m^{1-\eps}$ vertices in $U_1$, which forces us to choose that many vertices in $U_2$. Since each vertex in $U_2$ is connected to exactly two vertices in $U_3$ via AND gates, satisfying any choice of $|U_1|$ vertices in $U_2$ require at most $2|U_1|$ vertices in $U_3$. However, this is also a lower bound, since an arbitrary choice of $m^{1-\eps}$ vertices in $U_3$ (corresponding to choosing $k_1=m^{1-\eps}$ vertices in $G_1$) will result in a subgraph with expected number of edges $p k_1^2
\approx m^{-(1-\eps)+2(1-\eps)}=m^{1-\eps}$, and a Chernoff bound shows that no significantly smaller subgraph can have this many induced edges. This argument may be repeated now for all the subsequent layers, since each time we will end up with $\Omega(m^{1-\eps})$ vertices in $U_{2i-1}$ that need to be satisfied, which forces us to choose that many vertices in $U_{2i}$ (since for any pair of vertices in $U_{2i-1}$, the number of their common neighbors in $U_{2i}$ is zero and the corresponding layer of $U_{2i} \rightarrow U_{2i-1}$ is an OR layer), and a subgraph with that many edges in $G_i$ must have $\Omega(m^{1-\eps})$ vertices, which means that we will incur a set of at least this size in $U_{2i+1}$. Thus, the optimum is $\Theta(m^{1-\eps})$.

Now let us see what fractional solution survives the basic SDP and Sherali-Adams.
At first we know that for a goal of $k_1^*\geq\sqrt{n}$ vertices (in $U_3$, we have a fractional solution that represents as $k_1^*$-subgraph with $k_1^*\sqrt{pn}$ edges. Thus, to achieve our goal of fractional weight $m^{1-\eps}$ in $U_2$ (representing edges in the subgraph in $G_1$), we need to choose $k_1^*$ such that
$k_1^*\sqrt{pn}=m^{1-\eps}$ which implies that $k_1^*\approx\frac{m^{1-\eps}}{m^{-(1-\eps)/2+(2-\eps)/4}}=m^{1-5\eps/4}$.
Thus, in our fractional solution, LP weight $m^{1-5\eps/4}$ is distributed evenly among the vertices of $U_3$. We can repeat this argument, showing that $U_{2i-1}$ has LP weight $\approx m^{1-(3+i)\eps/4}$ and that to get an fractional solution with this edge weight in $G_i$, it suffices to set the fractional vertex weight to $k^*_i=m^{1-(4+i)\eps/4}$, as long as $k^*_i\geq\sqrt{n}$. That is, as long as 
$m^{1-(4+i)\eps/4}\geq m^{(2-\eps)/4}$ which corresponds to $i\leq\frac{2}{\eps}-3$.
Once $i\geq\frac{2}{\eps}-3$, we can use the log-density lower bound for obtaining fractional edge weight $k^*_{i-1}=m^{1-(3+i)/4}$, which sets the fractional vertex weight to $k^*_i=m^{2\alpha_i/(1-2\eps)}\approx n^{\alpha_i}$ such that 
$$m^{1-(3+i)\eps/4}=k^*_i n^{\alpha_i\beta} = k_i^*n^{\alpha_i\cdot\eps/(2-\eps)} \approx n^{\alpha_i}n^{\alpha_i\cdot\eps/(2-\eps)}=n^{2\alpha_i/(2-\eps)}\approx m^{\alpha_i}.$$ Recall that this bound holds as long as $\alpha_i\leq 1/2$. Indeed, at this point, since $i\geq\frac{2}{\eps}-3$, we get
$\alpha_i=1-\frac{(3+i)\eps}{4}\leq 1-\frac{2}{4}=\frac12$. From this point on, the sequence of $k^*_i$ keeps decreasing, so it is always at most $\sqrt{n}$, and we can always use the log-density lower bound. 

We now follow a similar calculation. If the relaxation places fractional weight $k^*_{i-1}=m^{\delta_{i-1}}$ on $U_{2i-1}$, then to get this fractional edge weight in $G_i$, we need fractional vertex weight $k^*_i=n^{\alpha_i}$ such that 
$m^{\delta_{i-1}}\approx n^{2\alpha_i/(2-\eps)}\approx m^{\alpha_i}$,
and so the fractional weight on $U_{2i+1}$ will be $k^*_i=n^{\delta_{i-1}}=m^{\delta_{i-1}\cdot (2-\eps)/2}$. Thus, we get a decrease of a factor of $(2-\eps)/2$ in the exponent at each step, and an additional $\log(1/(2\eps))/\log(2/(2-\eps))$ steps will bring us down to $m^{\eps}$. Thus, the integrality gap is at least $m^{1-2\eps}$ for MMSA$_t$ for any
$t\geq \frac{2}{\eps}-3+\frac{\log(1/(2\eps))}{\log(2/(2-\eps))}$.
Thus the integrality gap does tend to linear as $t$ increases, with a dependence which is at most $t=O((1/\eps)\log(1/\eps))$, or in other words, an integrality gap of $m^{1-O((\log t)/t)}$.

\bibliographystyle{abbrvnat}
\bibliography{rbsc}

\begin{thebibliography}{16}
\providecommand{\natexlab}[1]{#1}
\providecommand{\url}[1]{\texttt{#1}}
\expandafter\ifx\csname urlstyle\endcsname\relax
  \providecommand{\doi}[1]{doi: #1}\else
  \providecommand{\doi}{doi: \begingroup \urlstyle{rm}\Url}\fi

\bibitem[Abidha and Ashok(2022)]{abidha2022red}
V.~Abidha and P.~Ashok.
\newblock Red blue set cover problem on axis-parallel hyperplanes and other
  objects.
\newblock \emph{arXiv preprint arXiv:2209.06661}, 2022.

\bibitem[Alekhnovich et~al.(2001)Alekhnovich, Buss, Moran, and Pitassi]{ABMP01}
M.~Alekhnovich, S.~Buss, S.~Moran, and T.~Pitassi.
\newblock Minimum propositional proof length is np-hard to linearly
  approximate.
\newblock \emph{The Journal of Symbolic Logic}, 66\penalty0 (1):\penalty0
  171--191, 2001.

\bibitem[Ashok et~al.(2017)Ashok, Kolay, and Saurabh]{ashok2017multivariate}
P.~Ashok, S.~Kolay, and S.~Saurabh.
\newblock Multivariate complexity analysis of geometric red blue set cover.
\newblock \emph{Algorithmica}, 79\penalty0 (3):\penalty0 667--697, 2017.

\bibitem[Awasthi et~al.(2010)Awasthi, Blum, and Sheffet]{awasthi2010improved}
P.~Awasthi, A.~Blum, and O.~Sheffet.
\newblock Improved guarantees for agnostic learning of disjunctions.
\newblock In \emph{Proceedings of the Conference on Learning Theory}, pages
  359--367, 2010.

\bibitem[Carr et~al.(2000)Carr, Doddi, Konjevod, and Marathe]{carr2000red}
R.~D. Carr, S.~Doddi, G.~Konjevod, and M.~Marathe.
\newblock On the red-blue set cover problem.
\newblock In \emph{Proceedings of the Symposium on Discrete Algorithms}, pages
  345--353, 2000.

\bibitem[Chan and Hu(2015)]{chan2015geometric}
T.~M. Chan and N.~Hu.
\newblock Geometric red--blue set cover for unit squares and related problems.
\newblock \emph{Computational Geometry}, 48\penalty0 (5):\penalty0 380--385,
  2015.

\bibitem[Charikar et~al.(2016)Charikar, Naamad, and Wirth]{CNW16}
M.~Charikar, Y.~Naamad, and A.~Wirth.
\newblock On approximating target set selection.
\newblock In \emph{Proceedings of the International Workshop on Approximation,
  Randomization, and Combinatorial Optimization}, 2016.

\bibitem[Chlamt{\'a}{\v{c}} and Manurangsi(2018)]{chlamtavc2018sherali}
E.~Chlamt{\'a}{\v{c}} and P.~Manurangsi.
\newblock {Sherali--Adams} integrality gaps matching the log-density threshold.
\newblock \emph{Approximation, Randomization, and Combinatorial Optimization.
  Algorithms and Techniques}, 2018.

\bibitem[Chlamt{\'a}{\v{c}} et~al.(2016)Chlamt{\'a}{\v{c}}, Dinitz, Konrad,
  Kortsarz, and Rabanca]{chlamtac2016densest}
E.~Chlamt{\'a}{\v{c}}, M.~Dinitz, C.~Konrad, G.~Kortsarz, and G.~Rabanca.
\newblock The densest {$k$}-subhypergraph problem.
\newblock In \emph{Proceedings of the International Workshop on on
  Approximation, Randomization, and Combinatorial Optimization. Algorithms and
  Techniques}, 2016.

\bibitem[Chlamt{\'a}{\v{c}} et~al.(2017)Chlamt{\'a}{\v{c}}, Dinitz, and
  Makarychev]{chlamtavc2017minimizing}
E.~Chlamt{\'a}{\v{c}}, M.~Dinitz, and Y.~Makarychev.
\newblock Minimizing the union: Tight approximations for small set bipartite
  vertex expansion.
\newblock In \emph{Proceedings of the Symposium on Discrete Algorithms}, pages
  881--899, 2017.

\bibitem[Dinur and Safra(2004)]{dinur2004hardness}
I.~Dinur and S.~Safra.
\newblock On the hardness of approximating label-cover.
\newblock \emph{Information Processing Letters}, 89\penalty0 (5):\penalty0
  247--254, 2004.

\bibitem[Elkin and Peleg(2007)]{elkin2007hardness}
M.~Elkin and D.~Peleg.
\newblock The hardness of approximating spanner problems.
\newblock \emph{Theory of Computing Systems}, 41\penalty0 (4):\penalty0
  691--729, 2007.

\bibitem[Goldwasser and Motwani(1997)]{GM97}
M.~Goldwasser and R.~Motwani.
\newblock Intractability of assembly sequencing: Unit disks in the plane.
\newblock In \emph{Proceedings of WADS}, volume~97, pages 307--320, 1997.

\bibitem[Madireddy and Mudgal(2022)]{madireddy2022constant}
R.~R. Madireddy and A.~Mudgal.
\newblock A constant--factor approximation algorithm for red--blue set cover
  with unit disks.
\newblock \emph{Algorithmica}, pages 1--33, 2022.

\bibitem[Madireddy et~al.(2021)Madireddy, Nandy, and
  Pandit]{madireddy2021geometric}
R.~R. Madireddy, S.~C. Nandy, and S.~Pandit.
\newblock On the geometric red-blue set cover problem.
\newblock In \emph{Proceedings of the International Workshop on Algorithms and
  Computation}, pages 129--141, 2021.

\bibitem[Miettinen(2008)]{miettinen2008positive}
P.~Miettinen.
\newblock On the positive--negative partial set cover problem.
\newblock \emph{Information Processing Letters}, 108\penalty0 (4):\penalty0
  219--221, 2008.

\end{thebibliography}

\appendix

\section{Adapting and Applying our Algorithm to Partial Red-Blue Set Cover}\label{sec:partial-rbsc}
Let us now consider the variation in which we are given a parameter $\hat k$, and are only required to cover at least $\hat k$ elements in a feasible solution. The algorithm and analysis work with almost no change other than the following.

In the algorithm, the stopping condition of the loop is of course no longer once we have covered all blue elements, but once we have covered at least $\hat k$ of them.

A slightly more subtle change involves the analysis of the LP rounding in the final iteration. The notion of progress towards a certain approximation guarantee may not be valid if the ratio of red elements to blue elements covered is still as small as required, but the number of blue elements added is far more than we need. Rather than derandomize the rounding, one can show that it succeeds (despite this issue) with high probability. Let us briefly sketch the argument here.

First, note that  we can always preemptively discard any sets with more than $\opt$ red elements, and so we may assume that $r_\alpha\leq\opt$. 
 %As before, we use the method of conditional expectations to construct a set $\hat J$ while maintaining the non-negativity of the value
%$$\expec[|\Gamma_{R_\alpha}(J^*)|]\cdot|\Gamma_{B}(\hat J)| - \expec[|\Gamma_{B}(J^*)|]\cdot|\Gamma_{R_\alpha}(\hat J)|.$$
Suppose we need to cover an additional $k^*$ elements in order to reach the target of $\hat k$ blue elements total. %However, now we stop as soon as the set $\Gamma_{B}(\hat J)$ contains at least $k^*$ elements. We can now check that we do not accrue too many red elements: In the previous step (of the rounding) we had an appropriately bounded number of red elements, by the same analysis as before, while in the current (last) step, we may overshoot the required number of blue elements, though the number of additional red elements added is at most $r_\alpha\leq\opt$ red elements, which does not hurt our approximation guarantee.
Since our bound on $\expec[|\Gamma_{R_\alpha}(J^*)|]$ is a linear function of our bound on $\expec[|J^*\setminus J_+|]$, by a Chernoff bound we have $|\Gamma_{R_\alpha}(J^*)|=\tilde O(r_\alpha A)$ with all but exponentially small probability. On the other hand, $|\Gamma_{B}(J^*)|$ is always at most $|B|$, so by Markov, we have
$$\prob\left[|\Gamma_{B}(J^*)|\leq \frac{\expec[|\Gamma_{B}(J^*)|]}{2}\right]\leq \frac{|B|-\expec[|\Gamma_{B}(J^*)|]}{|B|-\expec[|\Gamma_{B}(J^*)|]/2}\leq 1-\frac{1}{2|B|}.$$ Thus, repeating the rounding a polynomial number of times (in a given iteration), with all but exponentially small probability we can find a set $\hat J\subseteq J$ that satisfies both $$|\Gamma_{B}(\hat J)|\geq \frac{\expec[|\Gamma_{B}(J^*)|]}{2}\qquad\text{and}\qquad |\Gamma_{R_{\alpha}}(\hat J)|]=\tilde O(r_\alpha A).$$ Now if $\expec[|\Gamma_{B}(J^*)|]\leq 2k^*$, then we have the required ratio and bound on the number of new red elements by the previous analysis. If $\expec[|\Gamma_{B}(J^*)|]> 2k^*$, then this will be the last iteration, as we will cover at least the required $k^*$ additional blue elements, and the number of red elements added at this final stage is at most $\tilde O(r_\alpha A)\leq \tilde O(A\cdot\opt)$, so we maintain the desired approximation ratio.

\section{Additional LP Constraints for MMSA\texorpdfstring{$_4$}{4}}\label{sec:LP}
The following is a complete list of lifted constraints that we use in addition to the basic LP relaxation for MMSA$_4$:
\begin{align*}
&X^{(h)}_j=X^{(j)}_h&\forall j\in J,\forall h\in S\\
&X^{(j)}_j=x_j&\forall j\in J\\
&\sum_{h\in\Gamma_S(i)}X_h^{(j)}\geq X_i^{(j)}&\forall j\in J\forall i\in R\\
&X_j^{(j)}\leq X_i^{(j)}&\forall j\in J\forall i\in\Gamma_R(j)\\
&0\leq X_a^{(j)}\leq x_j&\forall j\in J\forall a\in \{j\}\cup R\cup S\\
&\sum_{h'\in S}w_{h'}^{(h)}\leq \opt&\forall h\in S\\
&\sum_{\ell\in B}X_{\ell}^{(h)}\geq |B|/(\log k\log m)w_h&\forall h\in S\\
&X_{\ell}^{(h)}\leq \sum_{j\in\Gamma_{J}(\ell)}X_{\ell,j}^{(h)}\leq 2e\ln(2k)X_{\ell}^{(h)}&\forall h\in S\forall\ell\in B\\
&\Delta X_j^{(h)}\leq\sum_{\ell\in\Gamma_B(j)}X_{\ell,j}^{(h)}\leq 2\Delta X_j^{(h)}&\forall h\in S\forall j\in J\\
&0\leq X_{\ell,j}^{(h)}\leq X_j^{(h)},X_{\ell}^{(h)}\leq w_h&\forall h\in S\forall \ell\in B\forall j\in J\\
&\sum_{h'\in\Gamma_S(i)}X_{h'}^{(h)}\geq X_i^{(h)} &\forall h\in S\forall i\in R\\
%&0\leq w_h\leq y_i\leq 1 &\forall (i,h)\in E(R,S)\\
&X_j^{(h)}\leq X_i^{(h)} &\forall h\in S\forall (j,i)\in E(J,R)\\
&0\leq X_a^{(h)}\leq w_h&\forall h\in S\forall a\in B\cup J\cup R\cup S
\end{align*}

\end{document}